\documentclass[lettersize, journal]{IEEEtran}
\usepackage{array}
\usepackage{float}
\usepackage{caption}
\usepackage{color}

\usepackage{graphicx} \usepackage{float} %\usepackage{subfigure} 
\usepackage{subcaption}
\usepackage{cite}
\usepackage{amsmath,amssymb,amsfonts}
\usepackage{amsthm}
\usepackage{graphicx}
\usepackage{textcomp}
\usepackage{xcolor}
\usepackage{svg}
\usepackage{algorithm}
\usepackage{algorithmic}
\usepackage{mathtools}
\usepackage{lipsum}
\usepackage{bm}

\newtheorem{proposition}{Proposition}
\newtheorem{remark}{Remark}

\newtheorem{definition}{Definition}
\newtheorem{theorem}{Theorem}

\newtheorem{lemma}{Lemma}

\newtheorem{corollary}{Corollary}

\allowdisplaybreaks
\makeatletter
\def\ScaleIfNeeded{%
\ifdim\Gin@nat@ width>\linewidth \linewidth \else \Gin@nat@width
\fi } \makeatother
\graphicspath{{Fig/}}
\begin{document}
\title{Multi-Mode Pinching-Antenna Systems: An Inter-Mode-Interference-Free Perspective}
% \title{Exploiting Multi-Mode Pinching-Antenna Systems for ISAC: Physical Modeling and Optimization}
\author{Jingjing Zhao, Songtao Xue, Xidong Mu, Kaiquan Cai, and Zhiguo Ding,~\IEEEmembership{Fellow,~IEEE}
\thanks{J. Zhao, S. Xue and K. Cai are with the School of Electronics and Information Engineering, Beihang University, Beijing, 100191, China, and also with the State Key Laboratory of CNS/ATM, Beijing, 100191, China (e-mail:\{jingjingzhao, xuesongtao, caikq\}@buaa.edu.cn). X. Mu is with the Centre for Wireless Innovation (CWI), Queen's University Belfast, Belfast, BT3 9DT, U.K (e-mail: x.mu@qub.ac.uk). Z. Ding  is with the School of Electrical and Electronic Engineering, Nanyang Technological University, 639798, Singapore (e-mail:
zhiguo.ding@ntu.edu.sg).}}
\maketitle
\begin{abstract}
The physical model of multi-mode pinching-antenna systems (PASS) is proposed based on the coupled-mode theory. Within the considered model, multiple guided modes are simultaneously excited in a single waveguide and exploited as independent signal-bearing channels, thus providing additional \textit{modal degrees of freedom} for signal transmission. Under the local electromagnetic perturbations introduced by pinching antennas (PAs), the undesired guided-modes coupling is explicitly investigated, and the resulting \textit{inter-mode interference (IMI)} issue is revealed. The sufficient hardware-design condition for achieving the IMI-free regime is further established, for which the practical feasibility is validated through full-wave electromagnetic simulations. Then, a tractable signal model is derived for multi-mode PASS operating over the IMI-free regime. {To demonstrate the benefits of the proposed multi-mode PASS model, the integrated sensing and communications (ISAC) is studied as a representative application scenario, where the base station simultaneously communicates with a communication user and senses a target by using two guided modes.} A joint baseband and pinching beamforming optimization problem is formulated for the minimization of the Cramér--Rao bound for target localization, subject to the minimum communication rate requirement, under both continuous and discrete PAs activation cases. An alternating optimization-based algorithm is developed to address the formulated non-convex problem. For the baseband beamforming, a penalty-based successive convex approximation method is invoked. For the pinching beamforming, a particle swarm optimization algorithm and a two-sided matching algorithm are proposed for the continuous and discrete PAs activation cases, respectively. Numerical results obtained in the ISAC application scenario demonstrate the superiority of the proposed multi-mode PASS model over the single-mode PASS.
\end{abstract}

\section{Introduction}
Recently, pinching-antenna systems (PASS) have emerged as a promising flexible-antenna paradigm, owing to the capability of dynamically deploying dielectric elements, termed pinching antennas (PAs), along low-loss waveguides~\cite{10945421,pinching2}. Compared to conventional flexible-antenna systems such as fluid antennas~\cite{9264694}, movable antennas~\cite{10318061} and reconfigurable intelligent surfaces~\cite{marco}, the unique features of PASS are two-fold. Firstly, PASS can extend waveguides toward services area, allowing PAs to be deployed close to users, thus constructing line-of-sight (LoS)-dominant communication links and effectively reducing free-space path loss~\cite{pinching-noma}. Secondly, PASS apertures can be reconfigured by adding or releasing low-cost PAs, which facilitates a practical approach to realize more flexible multiple-input multiple-output (MIMO) transmissions~\cite{pinching-PLS}. 

The aforementioned benefits have stimulated growing research interests in PASS. For example, the authors of~\cite{11202577} investigated the data rate improvement brought by PASS through optimizing PAs positions, which was so-called the pinching beamforming. The authors of~\cite{pinching-arraygain} further characterized the achievable array gain obtained by PASS and derived the optimal number of PAs for realizing the maximum array gain. In~\cite{11263923}, the authors studied the transmit power minimization problem for PASS subject to users' minimum communication rate constraints. For the multi-user uplink communication scenario, the authors in~\cite{pinching-uplink} aimed at maximizing the minimum achievable rate of users, so as to obtain a balance between throughput and fairness. Exploiting the dynamic feature of wireless propagation environments, a novel PASS-based environment division multiple access scheme was proposed in~\cite{11488474}, where the intended-signal strength and multi-user interference were properly managed by adjusting PAs positions.  
{The flexible deployment and beam-shaping capabilities of PASS also make them attractive for integrated sensing and communications (ISAC), where communication and sensing functionalities share the same spectrum and hardware resources~\cite{CRB}. The authors of~\cite{Related-ISAC-LHC} studied the ISAC performance in multiple-waveguide PASS, where transmitting PAs and a uniform linear array (ULA) were adopted at the base station (BS) for ISAC signal transmission and echo signal reception, respectively. Moreover, it was demonstrated in~\cite{11212813} that PASS outperformed conventional MIMO systems in terms of communication rate under the radar signal-to-noise budget, through jointly optimizing baseband and pinching beamforming. Considering the discrete PAs activation case in~\cite{11599010}, a PASS-assisted ISAC framework was established under equal and proportional power radiation models. Moreover, 
By applying PASS in low-altitude scenarios, the authors of~\cite{11288076} designed a joint pinching beamforming, PAs activation states, and waveguide power allocation algorithm, with the aim of enhancing both sensing and communication performance.} 

\begin{figure*}[t]
    \centering
    \begin{subfigure}{0.322\linewidth}
		\centering	\includegraphics[width=1\linewidth]{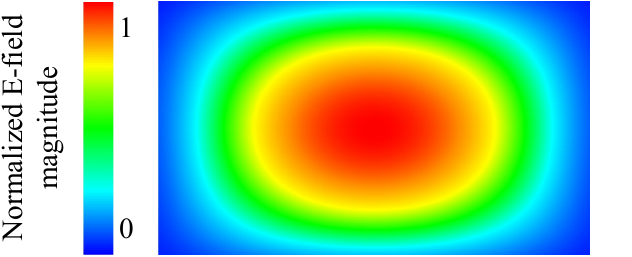}
        \captionsetup{margin={4.5mm,-4.5mm}}
		\caption{quasi-TE$_{00}$}
		\label{quasi-TE-00}
	\end{subfigure}
     \hspace{6mm}
	\centering
	\begin{subfigure}{0.28\linewidth}
		\centering
		\includegraphics[width=1.02\linewidth]{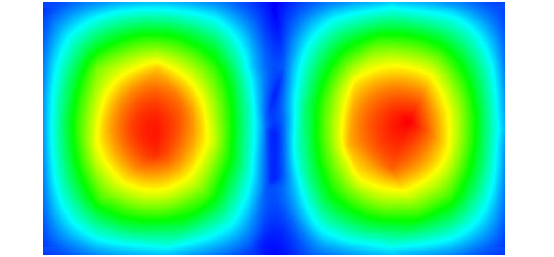}
		\caption{quasi-TE$_{10}$}
		\label{quasi-TE-10}
	\end{subfigure}
    \hspace{6mm}
	\centering
	\begin{subfigure}{0.26\linewidth}
		\centering
		\includegraphics[width=1\linewidth]{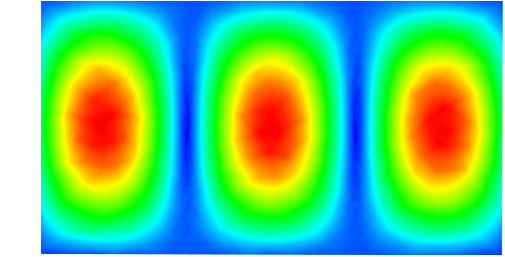}
        \captionsetup{margin={1.5mm,-1.5mm}}
		\caption{quasi-TE$_{20}$}
		\label{quasi-TE_20}%文中引用该图片代号
	\end{subfigure}
    
    \caption{Electric-field magnitude distributions of three typical quasi-TE modes over the cross section of a rectangular dielectric waveguide operating at $15$~GHz. The waveguide is made of Arlon CLTE-AT\textsuperscript{TM} with a relative permittivity of $3$ and a $28~\mathrm{mm}\times15~\mathrm{mm}$ cross section.}
    \label{fig:mode}
\end{figure*}

Existing studies on PASS have focused primarily on the simple single-mode architecture, where the electromagnetic-field evolution inside the waveguide is characterized by a single dominant guided mode. As a result, each waveguide can support only a single independent data stream due to the limited multiplexing capability. {This limitation may restrict the applicability of PASS in scenarios requiring multiple simultaneous data streams or higher multiplexing capabilities.}
For example, in PASS-enabled ISAC systems, sufficient transmit degrees of freedom (DoFs) are important for flexibly directing signal resources toward communication users (CUs) and sensing targets (STs). However, the multiplexing limitation of single-mode PASS results in a rank-one transmit covariance matrix, which may lead to a pronounced communication-sensing tradeoff. Previous works~\cite{Related-ISAC-LHC,11212813,11599010,11288076} proposed to employ multiple waveguides to support the transmission of multiple data streams, which inevitably incurs increased hardware complexity and deployment costs.

To overcome the above drawbacks of conventional single-mode PASS, in this work, we introduce the \textit{multi-mode PASS} architecture, where multiple guided modes in a single waveguide are exploited  as independent signal-bearing channels. In practice, as the transverse dimension of the dielectric waveguide increases, or as the operating frequency becomes sufficiently high, multiple propagating guided modes can be supported in a waveguide. Each guided mode corresponds to an eigensolution of Maxwell's equations subject to the material distribution and boundary conditions of the waveguide~\cite{Yariv1973CMT}. Fig.~\ref{fig:mode} demonstrates three typical transverse electric fields of a rectangular dielectric waveguide, i.e., quasi\footnote{``Quasi-TE'' indicates that the corresponding eigenmode exhibits a predominantly TE-like field distribution rather than being a strictly transverse-electric mode. This is because, the dielectric boundaries generally lead to hybrid electromagnetic modes in which the longitudinal components of both the electric and magnetic fields are nonzero.}-TE$_{00}$, quasi-TE$_{10}$, and quasi-TE$_{20}$, where quasi-$\mathrm{TE}_{00}$ is the fundamental mode and exhibits the simplest transverse field distribution, while quasi-$\mathrm{TE}_{10}$ and quasi-$\mathrm{TE}_{20}$ are higher-order modes with increased spatial variations.  

Note that, some initial works~\cite{xiaoxia-multimode,shaodan-multimode} have already started to exploit the additional multiplexing gain offered by multi-mode PASS. The authors of~\cite{xiaoxia-multimode} proposed a PA grouping scheme to support mode-domain multiplexing for multi-user communications, where the baseband and pinching beamforming were jointly optimized. In~\cite{shaodan-multimode}, a polarization-aware full-wave electromagnetic model was established to characterize the mode-dependent radiation patterns and polarization states. It is assumed in both~\cite{xiaoxia-multimode} and~\cite{shaodan-multimode} that, the orthogonality among guided modes remains preserved after introducing PAs. However, such an assumption does not hold for all conditions, given that a PA introduces a local electromagnetic perturbation to the waveguide. In this work, we explicitly account for the PA-induced undesired coupling among guided modes, and reveal the resulting issue of signal leakage and/or exchange among different guided modes, namely the \textit{inter-mode interference (IMI)}. We establish sufficient hardware-design conditions for achieving IMI-free multi-mode propagation in the waveguide, based on which a tractable multi-mode PASS signal model is given. We investigate ISAC as a representative application scenario to demonstrate the benefits of additional \textit{modal DoFs} brought by the proposed multi-mode PASS model. 

The main contributions of this paper are summarized as follows:

% To exploit the additional multiplexing DoFs offered by multiple guided modes in the waveguide, in this work, we investigate a multi-mode PASS-enabled ISAC framework. By simultaneously exciting multiple mutually orthogonal guided modes, multi-mode PASS provides additional modal DoFs for more flexible joint shaping of the communication and sensing signals. Compared with the single-mode counterpart, multi-mode PASS enlarges the feasible transmit signal space and relaxes the inherent rank limitation of single-mode transmission, thereby facilitating a more favorable performance tradeoff between communication and sensing. 
\begin{itemize}
\item  {We propose the physical model of multi-mode PASS, where multiple guided modes are simultaneously excited within a single waveguide and exploited as independent signal-bearing channels. Under the local perturbations introduced by PAs, we explicitly investigate the undesired guided-modes coupling, and reveal the resulting IMI issue. We further establish the hardware-design conditions for achieving the IMI-free regime, whose practical feasibility is validated through full-wave simulations using the High Frequency Structure Simulator (HFSS). Based on the IMI-free regime, we derive a tractable signal model for the multi-mode PASS.} 

% To preserve the independent in-waveguide propagation of multiple guided modes under the perturbation induced by PAs, we characterize the conditions for achieving IMI-free in-waveguide transmission and validate its practical feasibility through full-wave High Frequency Structure Simulator (HFSS) simulations. Based on the resulting IMI-free regime, we derive a tractable signal model for multi-mode PASS.
\item We study ISAC as a representative application scenario to demonstrate the benefits of additional modal DoFs offered by the proposed multi-mode PASS model. In the considered scenario, the BS simultaneously communicates with a CU and senses a ST by exploiting two guided modes in the waveguide. For both the continuous and discrete PAs activation cases, we formulate a joint baseband and pinching beamforming optimization problem that minimizes the Cram\'er-Rao bound (CRB) for target localization, subject to the minimum communication rate requirement.
\item We develop an AO-based scheme for iteratively solving the baseband and pinching beamforming subproblems. For the baseband beamforming, the penalty-based successive convex approximation (SCA) method is invoked. For the pinching beamforming, the particle swarm optimization (PSO) algorithm and matching algorithm are proposed to explore the available PAs positions space under the continuous and discrete activation cases, respectively.
\item {We conduct numerical results in the ISAC application scenario, to demonstrate the merits of the proposed IMI-free multi-mode PASS model. Numerical results unveil that the multi-mode PASS significantly improves the ISAC performance compared to the single-mode PASS, especially when the user deployment region expands or the minimum communication rate requirement becomes more stringent. 
This demonstrates that the additional modal DoFs are particularly beneficial in challenging ISAC scenarios, where the limited DoFs of the single-mode PASS become a major performance bottleneck.} 
%\textcolor{blue}{It is also shown that, with a sufficiently large number of candidate positions, the discrete activation case can achieve superior performance to the continuous activation case, by overcoming the PAs positions ordering constraint. }}

\end{itemize}

The rest of this paper is organized as follows. Section II introduces the multi-mode PASS physical modeling and the IMI analysis. Section III studies ISAC as a representative application scenario of the proposed IMI-free multi-mode PASS model, and formulates the joint baseband and pinching beamforming optimization problem. Section IV develops the AO-based algorithm for solving the resultant optimization problem. Section V provides numerical results to evaluate the performance of the proposed framework. Finally, section VI concludes the paper and summarizes the main findings.

\textit{Notations:}
Scalars, vectors, and matrices are denoted by lowercase letters, boldface lowercase letters, and boldface capital letters, respectively. $\mathbb{C}^{M\times N}$ and $\mathbb{R}^{M\times N}$ denote the sets of all $M\times N$ matrices with complex and real entries, respectively. $\mathbf{a}^*$, $\mathbf{a}^T$, and $\mathbf{a}^H$ represent the conjugate, transpose, and conjugate transpose of the vector $\mathbf{a}$, respectively. $\text{Re}\left\{\mathbf{A}\right\}$ and $\text{Im}\left\{\mathbf{A}\right\}$ denote the real and imaginary part of the matrix $\mathbf{A}$. $\mathrm{tr}\left(\mathbf{A}\right)$ and $\mathrm{rank}\left(\mathbf{A}\right)$ represent the trace and the rank of the matrix $\mathbf{A}$. $\mathbb{E}\left[\cdot\right]$ represents the statistical expectation. $\mathcal{CN}\left(a,b\right)$ and $\mathcal{U}[a,b]$ represent the Gaussian distribution with mean $a$ and variance $b$, and the uniform distribution over the interval $[a,b]$, respectively. $|\cdot|$ denotes the absolute value of a scalar. $\|\cdot\|$ represents the norm of a vector. $\mathbf{I}_M$ represents the $M\times M$ identity matrix. 

\section{Multi-Mode PASS Physical Modeling and Fundamental Analysis}
In this section, we first develop the multi-mode waveguide and PAs coupling model. Different from existing works that directly treat multiple guided modes as orthogonal channels, we explicitly investigate the PA-induced undesired mode coupling, and reveal the resulting IMI issue. We further establish the conditions under for achieving the IMI-free regime, which is validated through full-wave electromagnetic simulations. Finally, a tractable multi-mode PASS signal model is given. 

\subsection{Multi-Mode Waveguide-PA Coupling Model}
% For a longitudinally invariant
% waveguide, different guided eigenmodes are mutually orthogonal, which enables each of them to serve as an independent information-bearing channel~\cite{Mohanty2017,Elsawaf2024DualMode}. 
Suppose that a longitudinally invariant
waveguide deployed along the $x$-axis supports $M$ propagating guided modes, indexed by the set $\mathcal{M}=\left\{1, \dots, M\right\}$. The total electric field inside the $M$-mode waveguide is given by
\begin{equation}
\mathbf{E}_{\text{w}}(x,y, z) = \sum_{m=1}^{M} A_m(x)\, \mathbf{E}_m(y,z)\, e^{-j\beta_m x},
\label{eq:waveguide-field}
\end{equation}
where $A_m(x)$ denotes the slowly-varying complex amplitude of mode $m$,
$\mathbf{E}_m(y,z)$ represents the power-normalized transverse field distribution of mode $m$, and $\beta_m$ is the corresponding propagation constant. $\mathbf{E}_m(y,z)$ and $\beta_m$ are uniquely determined by the waveguide geometry and material parameters.

The localized PA can be regarded as a controlled perturbation that introduces coupling among the waveguide and PA modes, thus enabling power leakage into radiation channels. To facilitate a tractable coupled-mode formulation, we assume each PA supports only the fundamental mode. 
Then, the electric field inside the PA can be expressed as 
\begin{equation}
\mathbf{E}_{\text{p}}(x,y,z)
=
B(x)
\mathbf{E}_{\text{p}}(y,z)
e^{-j\beta_{\text{p}} x},
\label{eq:pa-field}
\end{equation}
where $B(x)$, $\mathbf{E}_{\text{p}}(y,z)$ and $\beta_{\text{p}}$ denote the slowly-varying complex amplitude, the power-normalized transverse field distribution, and the propagation constant of the PA mode, respectively. 

Since the PA is placed outside the dielectric waveguide and interacts with the guided modes through their evanescent fields, the resulting $(M+1)$-mode ($M$ guided modes and one PA mode) coupled system can be modeled based on the coupled-mode theory~\cite{Huang1984CMT} under the weak-perturbation assumption. Specifically, we collect the amplitudes of the $M$ guided modes and the PA mode into $\mathbf{f}(x)=\left[A_1(x), \dots, A_{M}(x), B(x)\right]^T$, then the coupled mode equations can be given by~\cite{Huang1984CMT}
\begin{equation}
    \label{eq:CMT}
    \frac{\mathrm{d}\mathbf{f}(x)}{\mathrm{d}x}=-j\mathbf{C}(x)\mathbf{f}(x),
\end{equation}
where
\begin{equation}
    \mathbf{C}(x)
    =
    \begin{bmatrix}
        \mathbf{C}_{\text{ww}}(x)
        &
        \mathbf{c}_{\text{wp}}(x)
        \\
        \mathbf{c}_{\text{pw}}(x)
        &
        0
    \end{bmatrix}\in\mathbb{C}^{(M+1)\times (M+1)}.
    \label{eq:interaction_coupling_matrix}
\end{equation}
In~\eqref{eq:interaction_coupling_matrix}, 
$\mathbf{C}_{\text{ww}}(x)\in\mathbb{C}^{M\times M}$ denotes the guided-modes coupling matrix, whose elements are given by
\begin{equation}
    \left[
        \mathbf{C}_{\text{ww}}(x)
    \right]_{mm'}
    =
    \begin{cases}
        \kappa_{mm'}
        e^{-j(\beta_{m'}-\beta_m)x},
        & m\neq m',\\
        0,
        & m=n.
    \end{cases}
    \label{eq:wg_coupling_matrix_entries}
\end{equation}
Here, $\kappa_{mm'}$ denotes the coupling coefficient between the waveguide modes $m$ and $m'$, given by
\begin{equation}
    \kappa_{mm'}
=
\frac{\omega\epsilon_0}{4}
\iint_{\Omega_p}
\Delta\epsilon_{\mathrm{p}}(y,z)
\mathbf E_m^{*}(y,z)
\cdot
\boldsymbol{E}_n(y,z)
\,\mathrm{d}y\mathrm{d}z,
\label{eq:kappa}
\end{equation}
where $\omega$ is the angular frequency corresponding to the carrier frequency, $\epsilon_{0}$ represents the vacuum permittivity, $\Omega_{\text{p}}$ is the coupling region, and 
$\Delta\epsilon_{\mathrm p}(y,z)$ denotes the transverse distribution 
of the permittivity perturbation, which is assumed 
to be invariant along the $x$-axis for a longitudinally uniform PA. Still referring to Eq.~\eqref{eq:interaction_coupling_matrix}, $\mathbf{c}_{\text{wp}}(x)\in\mathbb{C}^{M\times 1}$ denotes the guided modes-PA coupling vector, whose elements are given by
\begin{equation}
    \left[
        \mathbf{c}_{\mathrm{wp}}(x)
    \right]_m
    =
    \kappa_{m\text{p}}
    e^{-j(\beta_{\text{p}}-\beta_m)x},
    \label{eq:mode_pa_coupling_vector}
\end{equation}
where $\kappa_{m\text{p}}$ denotes the coupling coefficient between guided mode $m$ and the PA mode. 
For the reciprocal and lossless coupling between the guided modes and the PA mode, the coupling vectors satisfy $\mathbf{c}_{\text{pw}}(x)=\mathbf{c}_{\text{wp}}^{{H}}(x)$.

\begin{figure}[t]
    \centering
    \begin{subfigure}{\linewidth}
        \centering
        \includegraphics[width=0.92\linewidth]{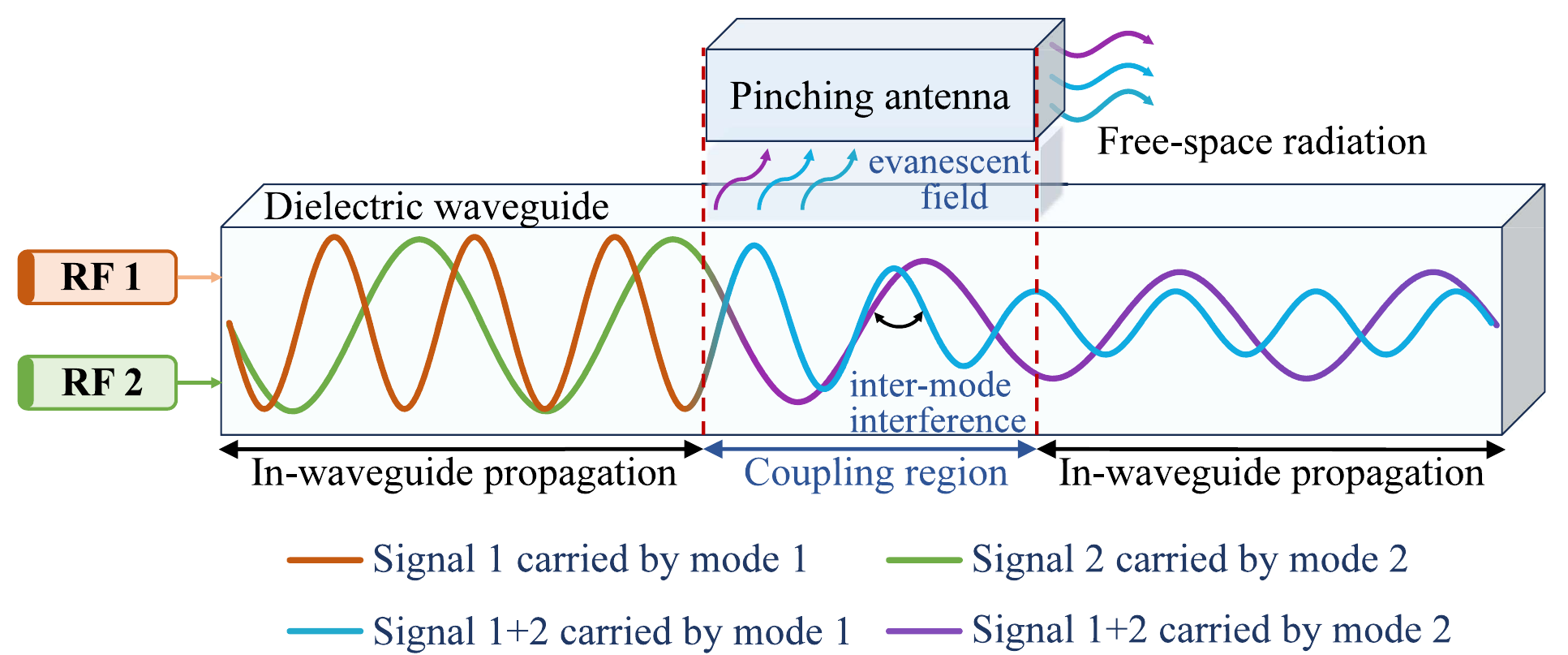}
        \caption{In-waveguide signal propagation with IMI.}
        \label{fig:couple_a}
    \end{subfigure}
    \par\vspace{4mm}
    \begin{subfigure}{\linewidth}
        \centering
        \includegraphics[width=0.93\linewidth]{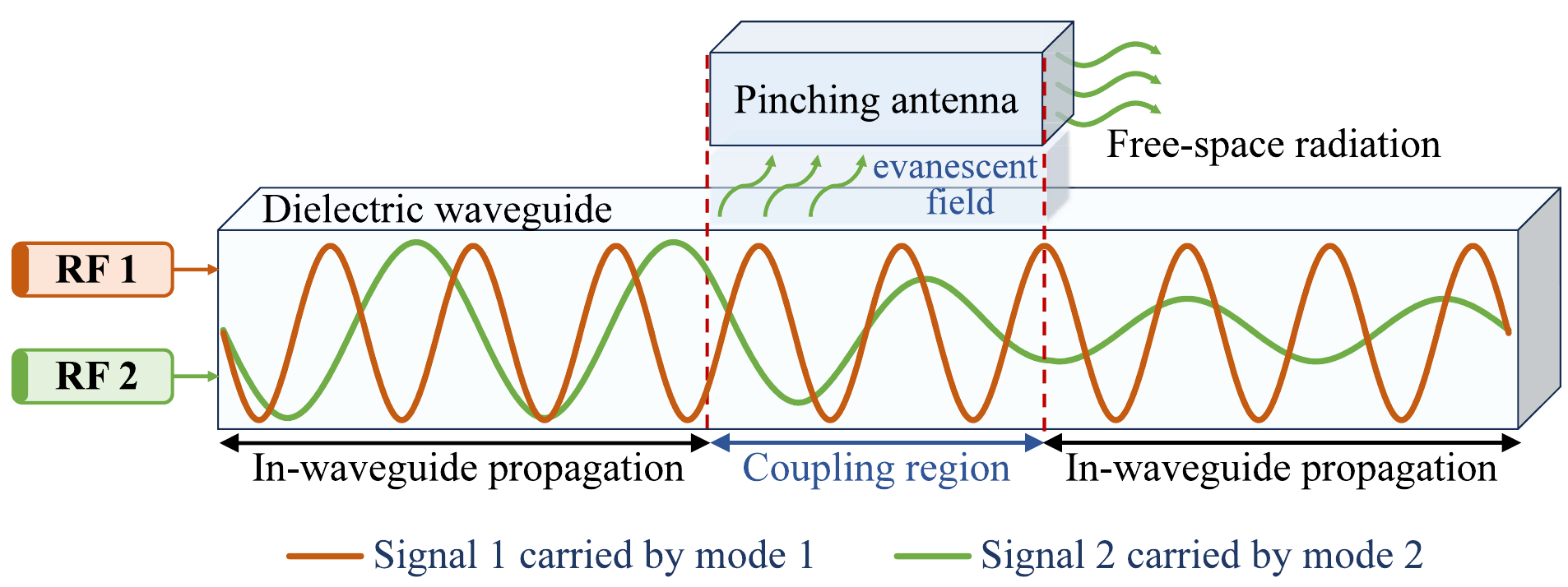}
        \caption{In-waveguide signal propagation without IMI.}
        \label{fig:couple_b}
    \end{subfigure}
    \caption{Demonstration of multi-mode PASS in-waveguide signal propagation.}
    \label{fig:sys_model}
\end{figure}
Based on Eq.~\eqref{eq:CMT}, we can obtain
\begin{equation}
    \mathbf{f}(x)
    =
    \mathbf{U}(x)\mathbf{f}(0),
    \label{eq:joint_modal_solution}
\end{equation}
where $\mathbf{U}(x)\in\mathbb{C}^{(M+1)\times (M+1)}$ is given by
\begin{align}
    \mathbf{U}(x)
    ={}&
    \mathbf{I}
    -
    j\int_{0}^{x}
    \mathbf{C}(\xi_1)\mathrm{d}\xi_1
    \notag\\
    &{}
    +
    (-j)^2
    \int_{0}^{x}
    \int_{0}^{\xi_1}
    \mathbf{C}(\xi_1)\mathbf{C}(\xi_2)\mathrm{d}\xi_2\,\mathrm{d}\xi_1
    +\cdots.
    \label{eq:dyson_expansion_phase_extracted}
\end{align}
{We can further partition $\mathbf{U}(x)$ as}
\begin{equation}
    \mathbf{U}(x)
    =
    \begin{bmatrix}
        \mathbf{T}_{\text{ww}}(x)
        &
        \mathbf{t}_{\text{wp}}(x)
        \\
        \mathbf{t}_{\text{wp}}(x)
        &
        T_{\text{pp}}(x)
    \end{bmatrix},
    \label{eq:partitioned_propagation_operator}
\end{equation}
where $\mathbf{T}_{\text{ww}}(x)\in\mathbb{C}^{M\times M}$ denotes the guided-modes complex amplitude transfer matrix, whose $(m,m')$-th entry is the amplitude transfer coefficient from guided mode $m'$ to $m$, $\mathbf{t}_{\text{wp}}(x)\in\mathbb{C}^{M\times 1}$ and $\mathbf{t}_{\text{pw}}(x)\in\mathbb{C}^{1\times M}$ represent the amplitude transfer vectors from the PA mode to guided modes and vice versa, respectively, and $T_{\text{pp}}(x)$ denotes the PA mode's self-transfer coefficient.
\subsection{Investigation of PA-Induced Inter-Mode Interference }
Without loss of generality, assume that the coupling starts at $x=0$ and ceases at $x=L_{\text{p}}$, with $L_{\text{p}}$ denoting the length of the PA. Given the initial condition of ${A}_m(0)
    =s_m, \forall m$ and $B(0)=0$, where $s_m$ is the signal excited into guided mode $m$ via its dedicated feed, we have
\begin{equation}
    A_m(x)
    =
    T_{mm}(x)s_m
    +   \sum_{\substack{m'\neq m}}
    T_{mm'}(x)s_{m'},
    \label{eq:modal_input_output_relation}
\end{equation}
\begin{equation}
    B(x)
    =
    \sum_{m=1}^{M}T_{\text{p}m}(x)s_m,
\end{equation}
where $T_{mm}(x)
    =
    \left[
        \mathbf{T}_{\mathrm{ww}}(x)
    \right]_{mm}$, $T_{mm'}(x)
    =
    \left[
        \mathbf{T}_{\mathrm{ww}}(x)
    \right]_{mm'}$ and $T_{\text{p}m}(x)=\left[
        \mathbf{t}_{\mathrm{\text{p}w}}(x)
    \right]_m$. Note that, $T_{mm}(x)$, $T_{mm'}(x)$, and $T_{\text{p}m}(x)$ account for not only the direct transition, but also all higher-order transitions mediated by other modes. For example, $T_{mm'}(x)$ can be expanded in the form of $T_{mm'}(x)=T_{mm'}^{(1)}(x)+T_{mm'}^{(2)}(x)+\cdots$, where $T_{mm'}^{(1)}(x)$ denotes the direct first-order transition from mode $m'$ to $m$, given by
\begin{equation}
\label{eq:m-n-direct-transition}
    T_{mm'}^{(1)}(x)
    =
    -j\kappa_{mm'}
    \int_{0}^{x}
    e^{-j(\beta_{m'}-\beta_m)\xi}
    \mathrm{d}\xi.
\end{equation}
For any $o>1$, $T_{mm'}^{(o)}(x)$ denotes the higher-order transmission from mode $m'$ to $m$. For example, the second-order contribution of $T_{mm'}^{(2)}$ is given by Eq.~\eqref{eq:second-order-CMT} shown on the top of the next page, which consists of the two-step indirect transition from mode $m'$ to $r$ and then to $m$, as well as the PA-mediated transition from mode $m'$ to PA and then to mode $m$. 
\begin{figure*}[ht]
    \begin{align}
        T_{mm'}^{(2)}(x)
=&\underbrace{{(-j)^2
\sum_{\substack{r=1\\ r\neq m,m'}}^{M}
\kappa_{mr}
\kappa_{rm'}\times
\int_{0}^{x} 
\int_{0}^{\xi_{1}} e^{-j(\beta_{r}-\beta_{m})\xi_{1}}
e^{-j(\beta_{m'}-\beta_{r})\xi_{2}}\mathrm{d}\xi_{2}\,\mathrm{d}\xi_{1}}}_{\text{amplitude transfer coefficient of path } m'\rightarrow r\rightarrow m}\nonumber\\
&\underbrace{+(-j)^2 \kappa_{m\text{p}}\kappa_{\text{p}m'}
\int_{0}^{x} 
\int_{0}^{\xi_{1}} 
e^{-j\left(\beta_{\text{p}}-\beta_{m}\right)\xi_{1}}
e^{-j\left(\beta_{m'}-\beta_{\text{p}}\right)\xi_{2}}\mathrm{d}\xi_{2}\,\mathrm{d}\xi_{1}}_{\text{amplitude transfer coefficient of path } m'\rightarrow \text{PA}\rightarrow m}.
    \label{eq:second-order-CMT}
    \end{align}
\hrulefill
\end{figure*}
$T_{mm}(x)$ and $T_{\text{p}m}(x)$ can be expanded in the similar way, for which the details are omitted here.
\begin{remark}
As observed in Eq.~\eqref{eq:modal_input_output_relation}, after the coupling among $(M+1)$ modes, the signal carried by each
guided mode $m$ contains not only its assigned signal $s_m$, but also signals carried by other guided modes, i.e., $s_{m'}, m'\neq m$. This phenomenon is referred to as the \textbf{inter-mode interference (IMI)}, as shown in Fig.~(\ref{fig:couple_a}).
\end{remark}

{Note that, if the PA exhibits non-negligible coupling with multiple
guided modes, it may serve as a common coupling path that redistributes
signal components among them, thereby leading to IMI, for which the second term on the right hand side of Eq.~\eqref{eq:second-order-CMT} demonstrates a good example. Therefore, to suppress IMI, each PA should exhibit {\textit{modal
selectivity}} with respect to one assigned guided mode, denoted by $q$. Specifically, the PA should couple predominantly to the assigned mode $q$, 
while keeping negligible coupling to other guided modes. Based on this modal-selectivity design, we next establish the conditions for achieving approximately IMI-free in-waveguide propagation.}
\begin{proposition}[Approximately IMI-Free In-Waveguide Propagation]
Define $\epsilon_{mm'}
\triangleq
\frac{2\left|\kappa_{mm'}\right|}
{\left|\beta_m-\beta_{m'}\right|}, m\neq m'$
,
$\epsilon_{\text{p}m}
\triangleq
\frac{2\left|\kappa_{\text{p}m}\right|}
{\left|\beta_{\text{p}}-\beta_m\right|}, m\neq q$, and $\epsilon
\triangleq
\max
\left\{
\max_{m\neq m'}\epsilon_{mm'},
\max_{m\neq q}\epsilon_{\text{p}m}
\right\}$.
Under the higher-order non-resonance assumption~\cite{Huang2009CMT} for undesired coupling paths, the IMI-to-signal power ratio satisfies 
\begin{equation}
    \frac{P_m^{\text{I}}}{P_{\text{s}}}=\mathcal{O}\left(\left(M-1\right)\epsilon^2\right), \quad \forall m,
\end{equation}
where $P_m^{\text{I}}$ and $P_{\text{s}}=\mathbb{E}\left[s_ms_m^*\right]$ denote the aggregated IMI power and the average signal power associated with mode $m$, respectively. Therefore, when $\left(M-1\right)\epsilon^2 \ll 1$, we have
\begin{equation}
    \frac{P_m^{\text{I}}}{P_{\text{s}}} \ll 1, \quad \forall m,
\end{equation}
which indicates that the $M$ guided modes can be regarded as
mutually independent signal-bearing channels approximately free of IMI. 
\label{prop:IMI-free}
\end{proposition}
\begin{proof}
    See Appendix A.
\end{proof}
According to Appendix~A, under the approximately IMI-free assumption, we have $T_{mm'}(x)=\mathcal{O}(\epsilon)\ll 1, m'\neq m$ and $T_{\text{p}m}(x)=\mathcal{O}(\epsilon)\ll 1, m\neq q$. Then we can approximate $A_m(x)$ and $B(x)$ as
\begin{equation}
    \label{eq:IMI-free-A}
    A_m(x)\simeq  
    \begin{cases}
s_m, & \text{if } m\neq q,\\
T_{qq}(x)s_q, & \text{if } m=q,
    \end{cases}
\end{equation}
\begin{equation}
    B(x)\simeq T_{\text{p}q}(x)s_q,
\end{equation}
where $T_{qq}(x)$ denotes the effective transition coefficient from guided mode $q$ to itself under the desired coupling with the PA mode, and $T_{\text{p}q}(x)$ is the effective transition coefficient from guided mode $q$ to the PA mode, satisfying $\left|T_{qq}(x)\right|^2+\left|T_{\text{p}q}(x)\right|^2=1$ under the lossless coupling assumption. In other words, the multi-mode coupling between the PA and the waveguide can be equivalently treated as a two-mode coupling between the PA mode and guided mode $q$, as shown in Fig.~(\ref{fig:couple_b}). {Therefore, referring to the derivations in the two-mode coupling case in~\cite{modeling-wzl}, for the special case of $\beta_{\text{p}}=\beta_q$ that enables the maximum free-space radiation power, $A_m(x)$ and $B(x)$ are give by}
\begin{equation}
    \label{eq:A}
    A_m(x)\simeq  
    \begin{cases}
s_m, & \text{if } m\neq q,\\
\operatorname{cos}(\left|\kappa_{\text{p}q}\right|x)s_q, & \text{if } m=q,
    \end{cases}
\end{equation}
\begin{equation}
    B(x)\simeq-j\operatorname{sin}(\left|\kappa_{\text{p}q}\right|x)s_q.
\end{equation}
\subsection{Feasibility Verification of the Proposed Multi-Mode PASS Model via Full-Wave Simulations}
\textbf{Proposition~\ref{prop:IMI-free}} indicates that, coupling coefficients and propagation constants of all waveguide guided modes and the PA mode should be carefully designed for guaranteeing the approximately IMI-free condition. In this regard, the transverse dimensions and the core-cladding refractive indices constitute the key design parameters~\cite{Mohanty2017}. To further verify the feasibility of the identified IMI-free in-waveguide propagation conditions, we conduct full-wave HFSS simulations with practical implementation considerations. Compared to a rotationally symmetric circular dielectric waveguide, the rectangular one exhibits more distinguishable transverse field patterns~\cite{Marcatili1969}, which is therefore adopted here for illustration. It is worthy noting that, the proposed multi-mode waveguide and PA coupling model and the associated IMI analysis are not restricted to rectangular waveguides and can be extended to other waveguide and PA geometries. 
As shown in Fig.~\ref{fig:hfss-simulation}, we illustrate the coupling of the PA quasi-$\mathrm{TE}_{00}$ mode with the waveguide quasi-$\mathrm{TE}_{00}$ mode and quasi-$\mathrm{TE}_{10}$ mode at $28~\mathrm{GHz}$. 
The dielectric waveguide is made of Arlon CLTE-AT\textsuperscript{TM} with a relative permittivity of $3$ and a rectangular cross section of $12~\mathrm{mm}\times1.5~\mathrm{mm}$. The PA is made of Arlon AD600\textsuperscript{TM} with a relative permittivity of $6.15$ and a rectangular cross section of $3.45~\mathrm{mm}\times1.5~\mathrm{mm}$. The distance between the PA and the dielectric waveguide is set to $0.05~\mathrm{mm}$. The waveguide quasi-$\mathrm{TE}_{00}$ and quasi-$\mathrm{TE}_{10}$ modes are with the propagation constants of $688.3~\mathrm{rad/m}$ and $602.8~\mathrm{rad/m}$, respectively, while the PA quasi-$\mathrm{TE}_{00}$ mode is with the propagation constant of $688.3~\mathrm{rad/m}$.

It can be observed that, the PA is located within the strong-field region of the guided quasi-$\mathrm{TE}_{00}$ mode, i.e., around the center of the waveguide, resulting in a substantial modal overlap and hence pronounced power coupling into the PA. Meanwhile, the PA is located close to the field null of the guided quasi-$\mathrm{TE}_{10}$ mode, leading to a significantly reduced modal overlap and negligible power transfer to the PA. The propagation-constant mismatch between the quasi-$\mathrm{TE}_{10}$ mode and the PA mode further suppresses this undesired coupling. The HFSS simulation results yield inter-mode transmission coefficients\footnote{We assign waveguide ports $1$ and $2$ to quasi-$\mathrm{TE}_{00}$ and quasi-$\mathrm{TE}_{10}$ modes, respectively. Accordingly, $S_{12}$ ($S_{21}$) denotes the amplitude
transferred from quasi-$\mathrm{TE}_{10}$ (quasi-$\mathrm{TE}_{00}$)
mode to quasi-$\mathrm{TE}_{00}$ (quasi-$\mathrm{TE}_{10}$) mode.} as $S_{12}=-44.18~\mathrm{dB}$ and $S_{21}=-45.12~\mathrm{dB}$, demonstrating negligible IMI.
These results demonstrate that, with proper waveguide and PA designs, the PA can selectively couple to the intended guided mode while suppressing undesired coupling to the other modes, and preserve approximately IMI-free in-waveguide propagation.

%Note that, since the $(M+1)$-mode coupling occurs in the evanescent field of the waveguide, where $\left\|\mathbf{E}_m\right\|, \forall m$ are much weaker than that inside the core, we can make the weak-perturbation assumption, i.e., $\kappa_{mm'}, \forall m\neq m'$ and $\kappa_{\text{p}m}, \forall m$ are sufficiently small. In this case, the propagation-constant mismatches play a dominant role in determining the multi-mode PASS performance. 

%\begin{enumerate}
    % \item \textbf{Sufficient guided-mode separation}: $\beta_m$ and $\beta_{m'}, \forall m\neq m'$ are sufficiently separated, so that undesired inter-guided-mode coupling is suppressed through    longitudinal phase cancellation. For example, in a rectangular    dielectric waveguide supporting the quasi-$\mathrm{TE}_{00}$ and    quasi-$\mathrm{TE}_{10}$ modes, the waveguide width can be appropriately    reduced to drive the    quasi-$\mathrm{TE}_{10}$ mode closer to cutoff and enlarge its    propagation-constant separation from quasi-$\mathrm{TE}_{00}$, while    ensuring that both modes remain guided~\cite{Mohanty2017}. 
    % \item \textbf{Mode-selective coupling by phase matching}: the fundamental mode of the PA is set to $\beta_{\text{p}}=\beta_{q}$ to enhance desired guided mode-PA coupling. Together with objective {1)}, the undesired guided mode-PA couplings are inherently suppressed. Taking the same example as that in objective 1), two different PA widths can be adopted to phase match     quasi-$\mathrm{TE}_{00}$ and quasi-$\mathrm{TE}_{10}$,     respectively~\cite{Mohanty2017}.
%\end{enumerate}
\begin{figure}[t]

    \centering
    \begin{subfigure}{0.75\linewidth}
        \centering
        \includegraphics[width=1\linewidth]{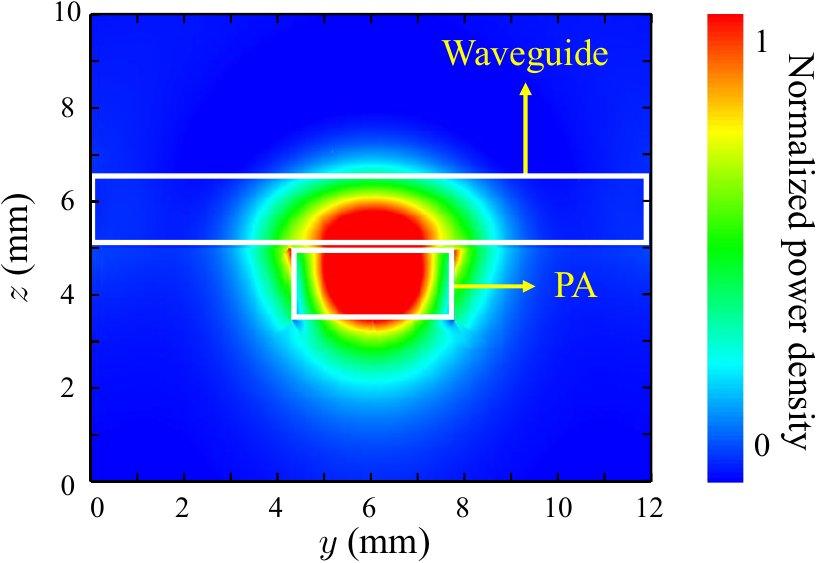}
        \caption{}
        \label{fig:quasi-TE00-PA}
    \end{subfigure}
    \centering
    \begin{subfigure}{0.75\linewidth}
        \centering
        \hspace*{0.001mm}
        \includegraphics[width=1.0\linewidth]{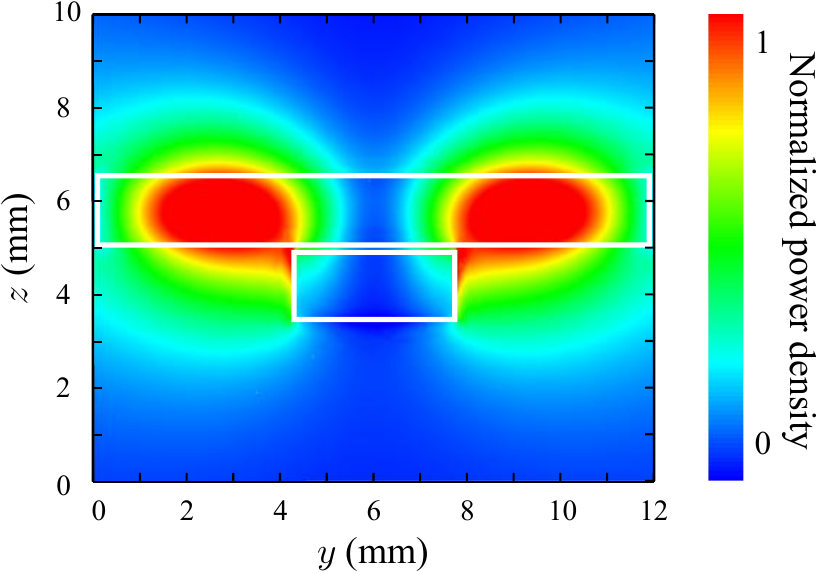}
        \caption{}
        \label{fig:quasi-TE10-PA}
    \end{subfigure}
     \hfill
    \caption{Normalized power density distributions over the waveguide-PA cross section at $x=L_{\text{p}}=6.545$~mm, for the coupling of the PA mode with the waveguide (a) quasi-TE$_{00}$, and (b) quasi-TE$_{10}$ modes, obtained from HFSS full-wave electromagnetic simulations.}
    \label{fig:hfss-simulation}
\end{figure}
   
\subsection{Multi-Mode PASS IMI-Free Signal Model}
Based on the above multi-mode PASS electromagnetic modeling, we now construct the IMI-free signal model from the communication perspective. Under the IMI-free regime, $M$ groups of PAs need to be activated, each radiating the signal carried by its selectively-coupled guided mode. Denote the set of all PAs and that of PAs selectively-coupled with guided mode $m$ by $\mathcal{N}=\left\{1, \dots, N\right\}$ and $\mathcal{N}_m$, respectively, satisfying $\mathcal{N}=\mathcal{N}_1\cup \dots \cup \mathcal{N}_M$ and $\mathcal{N}_{m}\cap \mathcal{N}_{m'}=\emptyset, \forall m\neq m'$. Then, the in-waveguide channel gain $g^{\text{w}}_{nm}$ from the feed point of mode $m$ to PA $n$ is given by
\begin{equation}
    g^{\text{w}}_{nm}=\prod_{ \substack{n'<n\\ n'\in\mathcal{N}_m}}T_{mm}^{n'}(L_{n'})e^{-j\beta_mx^{\text{feed},m}_{n}},
\end{equation}
where $T_{mm}^{n'}(L_{n'})$ denotes the effective transition coefficient from guided mode $m$ to itself after the selective coupling with the $n'$-th PA mode, with $L_{n'}$ representing the length of the $n'$-th PA, and
$x^{\text{feed},m}_{n}$ is the in-waveguide propagation distance of $s_m$ from its dedicated feed point to the coupling region of PA $n$. Moreover, denote the channel gain of PA $n$ for coupling $s_m$ to its radiating end by $g^{\text{p}}_{nm}$, which is given by
\begin{equation}
    g^{\text{p}}_{nm}=
    \begin{cases}
0, & \text{if } n\notin \mathcal{N}_m,\\
T_{nm}\left(L_{n}\right)e^{-j\beta_{n}L_{n}}, &\text{if } n\in\mathcal{N}_m,
    \end{cases}
\end{equation}
where $T_{nm}\left(L_{n}\right)$ is the effective transition coefficient from mode $m$ to the $n$-th PA mode, and $\beta_n$ is the propagation constant of the  $n$-th PA mode. Then, the signal received in the free space is given by 
\begin{equation}
    y_{\text{r}}=\sum_{n\in\mathcal{N}}\sum_{m\in\mathcal{M}}h_{n\text{r}}{g}_{nm}^{\text{p}}g_{nm}^{\text{w}}s_m+n_{\text{r}},
\end{equation}
where $h_{n\text{r}}=\frac{\eta e^{-j\beta_0d_{n\text{r}}}}{
    d_{n\text{r}}}$ denotes the free-space channel gain from PA $n$ to the receiving point, with $\eta$, $\beta_0$ and $d_{n\text{r}}$ representing the reference channel amplitude gain at a distance of $1$~m, the free-space propagation constant, and the distance between PA $n$ and the receiving point, respectively. Moreover, $n_{\text{r}}$ denotes the additive white Gaussian noise (AWGN),

Define $g_{nm}\triangleq g_{nm}^{\text{p}}g_{nm}^{\text{w}}$. If $n\notin \mathcal{N}_m$, we have $g_{nm}=0$. If $n\in\mathcal{N}_m$, for the special case of $\beta_n=\beta_m, \forall n\in\mathcal{N}_m$, we have a more tractable expression of $g_{nm}$ as
\begin{align}
    g_{nm}=-j\prod_{ \substack{n'<n\\ n'\in\mathcal{N}_m}}\operatorname{cos}\left(\phi_{n'm}\right)\operatorname{sin}\left(\phi_{nm}\right)e^{-j\beta_m\tilde{x}^{\text{feed},m}_{n}},
\end{align}
where $\phi_{nm}=\left|\kappa_{nm}\right|L_n$ and $\tilde{x}^{\text{feed},m}_{n}= {x}^{\text{feed},m}_{n}+ L_n$. In practice, the longitude length $L_n$ of each PA can be tailored to realize different power radiation ratios. For example, two simple power radiation models, i.e., equal and proportional power models, are proposed in~\cite{modeling-wzl}. 
%Although the power radiation control can bring in extra design flexibility, it also causes additional implementation complexity. As such, this work considers the fixed power radiation ratio at each PA. 

\section{Application of Multi-Mode PASS to ISAC- 
System Model and Problem Formulation}
{Based on the proposed physical and signal models of IMI-free multi-mode PASS in Section II, we next study ISAC as a representative application scenario to demonstrate the benefits of the additional modal DoFs. For ISAC, the communication and sensing functionalities generally compete for the available transmit resources, thus imposing demanding multiplexing requirements. By allowing multiple guided modes within a single waveguide to carry independent signals, the multi-mode PASS provide additional flexibilities for balancing these two functionalities. Accordingly, this section develops a multi-mode PASS-enabled ISAC framework to validate the merits of the proposed multi-mode PASS model.}
\vspace{-0.3cm}
\subsection{System Model}
\begin{figure}
    \centering
\includegraphics[width=0.9\linewidth]{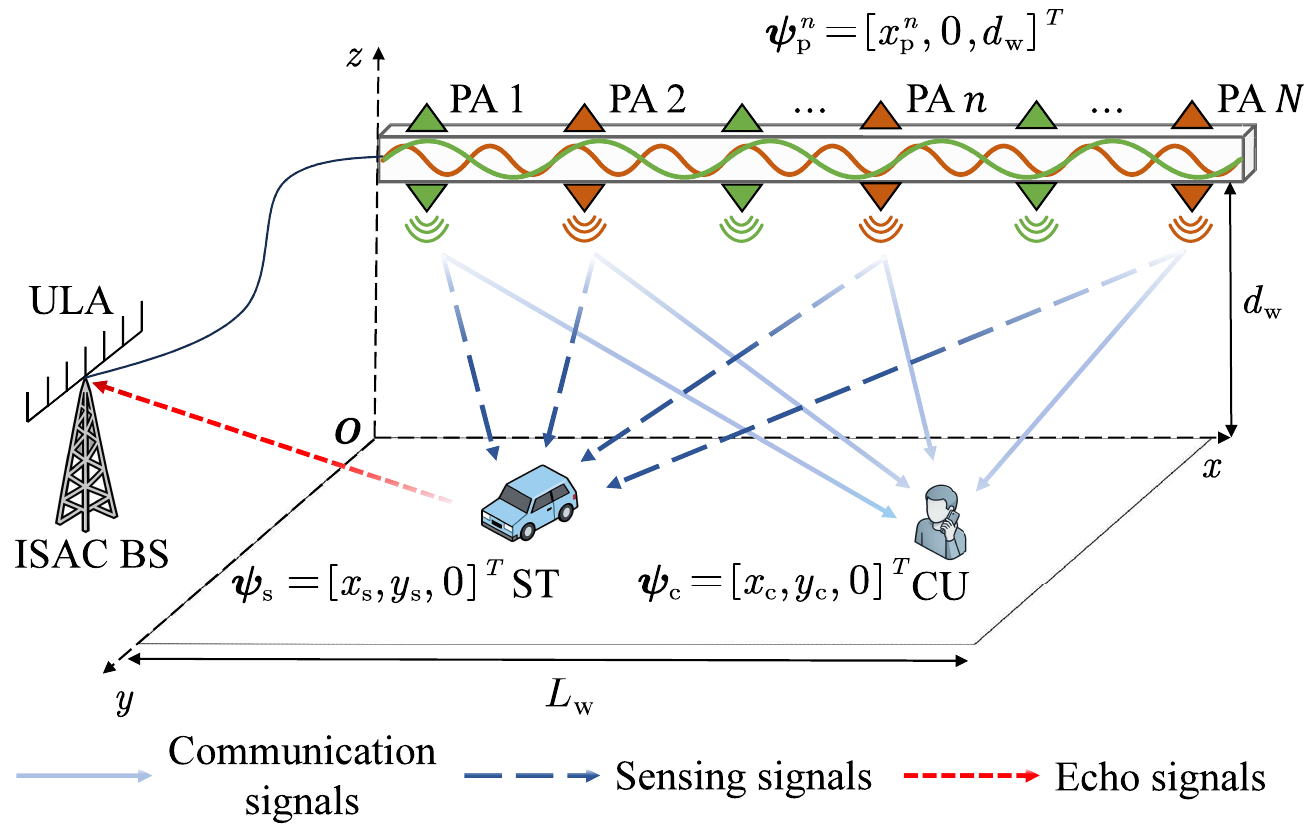}
    \caption{System model of the multi-mode PASS in an ISAC scenario.}
    \label{fig:sys_model}
\end{figure}
We consider the multi-mode PASS-enabled ISAC scenario as shown in Fig.~\ref{fig:sys_model}, where the BS is equipped with a dielectric waveguide for simultaneously transmitting information and sensing signals to a single-antenna CU and a ST, respectively. Considering the minimum spatial multiplexing requirement, we assume that the waveguide supports two guided modes, indexed by $\left\{1,2\right\}$.
The BS employs a ULA with $K$ antennas to receive the reflected echo signals for the target localization, where the antennas are with the half-wavelength spacing. We establish a three-dimensional (3D) Cartesian coordinate system, where the positions of CU and ST are denoted by $\boldsymbol{\psi}_{\text{c}}=\left[x_{\text{c}},y_{\text{c}},0\right]^T$ and $\boldsymbol{\psi}_{\text{s}}=\left[x_{\text{s}},y_{\text{s}},0\right]^T$, respectively.
The transmit waveguide with length $L_{\text{w}}$ is aligned parallel to the $x$-axis with a height of $d_{\text{w}}$. Without loss of generality, we assume that the dedicated feeds for different guided modes are located at the
same longitudinal position of $x_{\text{p}}^{\text{feed}}=0$. {We assume that the sets of PAs selectively-coupled with guided modes $1$ and $2$, i.e., $\mathcal{N}_1$ and $\mathcal{N}_2$, are predetermined.} {The position of PA $n$ within $\mathcal{N}_m$ is denoted by $\boldsymbol{\psi}_{\text{p}}^{m,n}=\left[x^{m,n}_{\text{p}},0,d_{\text{w}}\right]^T$. With respect to the complete PA set $\mathcal{N}=\mathcal{N}_1\cup\mathcal{N}_2$, the $x$-axis coordinates of all PAs are collected into $\mathbf{x}_{\text{p}}=\left[\mathbf{x}_{\text{p}}^1, \mathbf{x}_{\text{p}}^2\right]$, with $\mathbf{x}_{\text{p}}^m=\left[x_{\text{p}}^{m,n}\right]_{n\in\mathcal{N}_m}, \forall m\in\left\{1,2
\right\}$. For notational simplicity, the mode index $m$ is omitted, and $x_{\text{p}}^{n}$ is used to denote the position of PA $n$ hereafter.
In practice, the PA allocation among different modes provides additional design flexibility and introduces a new optimization dimension. However, this problem is beyond the scope of this work and will be investigated in future studies.} 

%$\left[x_{\text{p}}^{m,1},\dots,x_{\text{p}}^{m,N}\right]$, with $m\in\left\{1,2\right\}$ for each $x_{\text{p}}^{m,n}$, determined by the preconfigured PA allocation. 

{We consider two PAs activation cases, i.e., continuous and discrete activation. The continuous activation case is achieved by PAs mounted on sliders and driven by motors~\cite{tutorial-lyw}. The feasible set of $\mathbf{x}_{\text{p}}$ in the continuous activation case is given by 
\begin{equation}
    \label{eq:c-set}
    \mathcal{P}_{\text{c}}=\left\{\mathbf{x}_{
    \text{p}}\in\mathbb{R}^{N}\left|
0<x_{\text{p}}^{n}\leq L_{\text{w}},\left|x_{\text{p}}^{n}-x_{\text{p}}^{n'}\right|\geq\Delta,\forall n\neq n'
\right.
    \right\},
\end{equation}
where $\Delta$ is the minimum spacing among PAs for avoiding PAs coupling.
For the discrete activation case, the electronic control enables selection of pre-deployed PAs at candidate positions, leading to the feasible set of $\mathbf{x}_{\text{p}}$ as
\begin{equation}
    \mathcal{P}_{\text{d}}=\left\{\mathbf{x}_{
    \text{p}}\in\mathbb{R}^{N}\left|x_\text{p}^n\in\left\{p_1,\dots,p_A\right\}, x_\text{p}^n\neq x_\text{p}^{n'}, \forall n\neq n'
    \right.
    \right\},
\end{equation}
where $p_a$ represents a candidate position along the $x$-axis, satisfying $\left|p_{a}-p_{a'}\right|\geq \Delta, \forall a\neq a'$.} 

Before sending the communication and sensing signals to the RF chain, the baseband beamforming is performed to further enhance the multiplexing gain. Specifically, let $T$ denote the number of time-domain snapshots, then the baseband signal $\mathbf{S}\triangleq\left[\mathbf{s}_1,\mathbf{s}_2\right]^T\in\mathbb{C}^{2\times T}$ is given by
\begin{equation}
    \mathbf{S}=\mathbf{w}\mathbf{c}+\mathbf{Z}, 
\end{equation}
where $\mathbf{w}\in\mathbb{C}^{2\times1}$ denotes the baseband beamforming vector for the communication signal $\mathbf{c}\triangleq \left[c_1,\dots,c_T\right]$, and $\mathbf{Z}\triangleq\left[\mathbf{z}_1,\dots,\mathbf{z}_{T}\right]\in\mathbb{C}^{2\times T}$ denotes the dedicated sensing signal. Suppose that $\mathbf{z}_t\overset{\mathrm{i.i.d.}}{\sim} \mathcal{CN}\left(\mathbf{0},\mathbf{R}_{\mathbf{z}}\right)$, with covariance matrix $\mathbf{R}_{\mathbf{z}}=\mathbb{E}[\mathbf{z}_t\mathbf{z}_t^H]$. Given that the communication and sensing signals are uncorrelated, {the covariance matrix of the overall baseband signal is given by $\mathbf{R}_{\mathbf{s}}=\mathbf{w}\mathbf{w}^H+\mathbf{R}_{\mathbf{z}}$.}

Under the condition of IMI-free in-waveguide propagation, the signal received at the CU is given by
\begin{align}
\label{eq: receive communication signal}
\mathbf{y}_{\text{c}}&\left(\mathbf{x}_{\text{p}},\mathbf{w},\mathbf{R}_{\mathbf{z}}\right)=\sum_{m=1}^2\mathbf{h}^H_{\text{c}}\left(\mathbf{x}_{\text{p}}\right)\mathbf{g}_{m}\left(\mathbf{x}_{\text{p}}\right)\mathbf{S}[m,:]+\mathbf{n}_{\text{c}}\notag\\
&=\mathbf{h}_{\text{c}}^H\left(\mathbf{x}_{\text{p}}\right)\mathbf{G}\left(\mathbf{x}_{\text{p}}\right)\mathbf{w}\mathbf{c}+\mathbf{h}_{\text{c}}^H\left(\mathbf{x}_{\text{p}}\right)\mathbf{G}\left(\mathbf{x}_{\text{p}}\right)\mathbf{Z}+\mathbf{n}_{\text{c}},
\end{align}
where $\mathbf{h}_{\text{c}}\left(\mathbf{x}_{\text{p}}\right)\triangleq \left[h_{1\text{c}}, \dots, h_{N\text{c}}\right]^H\in\mathbb{C}^{N\times 1}$ with $h_{n\text{c}}=\frac{\eta e^{-j\beta_0d_{n\text{c}}}}{
    d_{n\text{c}}}$ and $d_{n\text{c}}=\left\|\boldsymbol{\psi}_{\text{c}}-\boldsymbol{\psi}_{\text{p}}^n\right\|$, $\mathbf{g}_m\left(\mathbf{x}_{\text{p}}\right)\triangleq \left[g_{1m},\dots,g_{Nm}\right]^T\in\mathbb{C}^{N\times 1}$, $\mathbf{G}\left(\mathbf{x}_{\text{p}}\right)\triangleq\left[\mathbf{g}_1,\mathbf{g}_2\right]\in\mathbb{C}^{N\times 2}$, and $\mathbf{n}_{\text{c}}\in\mathbb{C}^{1\times T}\sim\mathcal{CN}\left(\mathbf{0},\,\sigma_{\text{c}}^2\mathbf{I}_T\right)$ is the AWGN at the CU , with $\sigma_{\text{c}}^{2}$ representing the noise power. Accordingly, the achievable data rate at the CU is given by  
\begin{align}
\label{eq: communication rate CU}
&R_{\text{c}}\left(\mathbf{x}_{\text{p}},\mathbf{w},\mathbf{R}_{\mathbf{z}}\right)\notag\\
&=\log_2\left(1+\frac{\left|\mathbf{h}_{\text{c}}^H\left(\mathbf{x}_{\text{p}}\right)\mathbf{G}\left(\mathbf{x}_{\text{p}}\right)\mathbf{w}\right|^2}{\mathbf{h}_{\text{c}}^H\left(\mathbf{x}_{\text{p}}\right)\mathbf{G}\left(\mathbf{x}_{\text{p}}\right)\mathbf{R}_{\mathbf{z}}\mathbf{G}^H\left(\mathbf{x}_{\text{p}}\right)\mathbf{h}_{\text{c}}\left(\mathbf{x}_{\text{p}}\right)+\sigma_{\text{c}}^2}\right).
\end{align}

{Assuming that reflections from the CU are effectively suppressed through existing clutter suppression techniques~\cite{Skolnik2001Radar}, the reflected echo signal received at the BS is given by}
\begin{equation}
\label{eq:Y_r}
\mathbf{Y}_{\text{s}}\left(\mathbf{x}_{\text{p}},\mathbf{w},\mathbf{R}_{\mathbf{z}}\right)=\varrho\mathbf{a}\left(\theta_{\text{b}}\right)\mathbf{h}_{\text{s}}^H\left(\mathbf{x}_{\text{p}}\right)\mathbf{G}\left(\mathbf{x}_{\text{p}}\right)\mathbf{S}+\mathbf{N}_{\text{b}},
\end{equation}
where $\varrho$ is the amplitude response, and $\mathbf{a}\left(\theta_{\text{b}}\right)$ is the steering vector of the ULA, with $\theta_{\text{b}}$ denoting the azimuth angle to the ST. Moreover, $\mathbf{h}_{\text{s}}\left(\mathbf{x}_{\text{p}}\right)\triangleq \left[h_{1\text{s}}, \dots, h_{N\text{s}}\right]^H\in\mathbb{C}^{N\times 1}$, with $h_{n\text{s}}=\frac{\eta e^{-j\beta_0d_{n\text{s}}}}{
    d_{n\text{s}}}$ and $d_{n\text{s}}=\left\|\boldsymbol{\psi}_{\text{s}}-\boldsymbol{\psi}_{\text{p}}^n\right\|$, and $\mathbf{N}_{\text{b}}\in\mathbb{C}^{K\times T}$ represents the AWGN with each element following $\mathcal{CN}\left(0,\sigma_{\text{s}}^2\right)$. Focusing on ST localization, we employ the CRB associated with $(x_{\text{s}},y_{\text{s}})$ to quantify the sensing performance, which is given by~\cite{CRB}
\begin{equation}
\label{eq:CRB close form}
\mathrm{CRB}_{\boldsymbol{\psi}_{\text{s}}}\left(\mathbf{x}_{\text{p}},\mathbf{w},\mathbf{R}_{\mathbf{z}}\right)=\left(\mathbf{J}_{11}-\mathbf{J}_{12}\mathbf{J}_{22}^{-1}\mathbf{J}_{12}^T\right)^{-1},
\end{equation}
where $\mathbf{J}_{11}$, $\mathbf{J}_{12}$, $\mathbf{J}_{22}$ $\in\mathbb{R}^{2\times 2}$ are blocks of the fisher information matrix (FIM), for which the detailed expressions are given in Appendix~\ref{app:FIM}. 
The diagonal elements of $\mathrm{CRB}_{\boldsymbol{\psi}_{\text{s}}}\left(\mathbf{x}_{\text{p}},\mathbf{w},\mathbf{R}_{\mathbf{z}}\right)$ provide lower bounds on the mean squared errors of unbiased estimators of $x_{\mathrm{s}}$ and $y_{\mathrm{s}}$, respectively.
\subsection{Problem Formulation}
Our objective is to minimize the trace of the CRB matrix by jointly optimizing the PAs positions and baseband beamforming, subject to the CU's minimum data rate constraint.
The corresponding optimization problem is formulated as follows:
\begin{subequations}
\label{eq:optimization_problem}
\begin{equation}
\label{eq:objective_function}    \min_{\mathbf{x}_{\text{p}}, \mathbf{w},\mathbf{R}_{\mathbf{z}}}\mathrm{tr}\left(\mathrm{CRB}_{\boldsymbol{\psi}_{\text{s}}}\left(\mathbf{x}_{\text{p}},\mathbf{w},\mathbf{R}_{\mathbf{z}}\right)\right),
\end{equation}
% \vspace{-0.4cm}
\begin{equation}
\label{eq:constraint_p}
  {\rm{s.t.}}\ \ \mathrm{tr}\left(\mathbf{w}\mathbf{w}^H+\mathbf{R}_{\mathbf{z}}\right)\leq {P}_{\text{max}},
\end{equation}
\begin{equation}
\label{eq:constraint_qos}
 R_{\text{c}}\left(\mathbf{x}_{\text{p}},\mathbf{w},\mathbf{R}_{\mathbf{z}}\right)\geq R_{\text{min}},
\end{equation}
\begin{equation}
\label{eq:se_constraint0}
\mathbf{R}_{\mathbf{z}}\succeq 0,
\end{equation}
% \vspace{-0.4cm}
%\begin{equation}
%\label{eq:constraint_Delta}
%\textcolor{red}
%\end{equation}
\begin{equation}
\label{eq:constrait_x}
     \mathbf{x}_{\text{p}} \in \mathcal{P}_{\text{c}}/ \mathbf{x}_{\text{p}} \in\mathcal{P}_{\text{d}},
\end{equation}
\end{subequations}
where \eqref{eq:constraint_p} restricts the BS transmit power not to exceed the maximum value of $P_{\text{max}}$. \eqref{eq:constraint_qos} gives the constraint on the minimum communication rate of $R_{\text{min}}$. \eqref{eq:se_constraint0} ensures that $\mathbf{R}_{\mathbf{z}}$ is positive semidefinite. % $\eqref{eq:constraint_Delta}$ is the minimum spacing constraint among PAs to avoid antenna coupling, where $\Delta$ denotes the minimum PA spacing. 
\eqref{eq:constrait_x} gives constraints on feasible PAs positions for continuous/discrete activation cases.

\section{{ Application of Multi-Mode PASS to ISAC- Joint Baseband and Pinching Beamforming}}
In this section, we employ the AO framework to iteratively optimize the baseband and pinching beamforming subproblems. For given $\mathbf{x}_{\text{p}}$, $\left\{\mathbf{w}, \mathbf{R}_{\mathbf{z}}\right\}$ is optimized based on the penalty-based SCA method. For given $\left\{\mathbf{w}, \mathbf{R}_{\mathbf{z}}\right\}$, $\mathbf{x}_{\text{p}}$ is optimized by invoking the PSO and matching theory frameworks, for the continuous and discrete activation schemes, respectively.
\subsection{Baseband Beamforming}
To begin with, $\mathbf{w}$ and $ \mathbf{R}_{\mathbf{z}}$ are optimized with given pinching beamforming $\mathbf{x}_{\text{p}}$. We first introduce the auxiliary matrix {$\mathbf{U}\in\mathbb{R}^{2\times 2}$} to handle the non-convexity of the objective function~\eqref{eq:objective_function_wr}. Specifically, ~\eqref{eq:optimization_problem} can be rewritten as
\begin{subequations}
\label{eq:optimization_problem_wr_2}
\begin{equation}
\label{eq:objective_function_wr_2}    \min_{\mathbf{w},\mathbf{R}_{\mathbf{z}},\mathbf{U}} \mathrm{tr}\left(\mathbf{U}^{-1}\right),
\end{equation}
\begin{equation}
\label{eq:constraint_J}
{\rm{s.t.}} \ \ \begin{bmatrix}\mathbf{J}_{11}-\mathbf{U}&\mathbf{J}_{12}\\\mathbf{J}_{12}^T&\mathbf{J}_{22}\end{bmatrix}\succeq\mathbf{0},
\end{equation}
\begin{equation}
\label{eq:se_constraint2}
\mathbf{U}\succ0,
\end{equation}
\begin{equation}
\eqref{eq:constraint_p}-\eqref{eq:se_constraint0}.
\notag
\end{equation}
\end{subequations}
The problem~\eqref{eq:optimization_problem_wr_2} is still non-convex due to~\eqref{eq:constraint_p} and \eqref{eq:constraint_qos}. To address this, we introduce $\mathbf{W}=\mathbf{w}\mathbf{w}^H$, and the optimization problem~\eqref{eq:optimization_problem_wr_2} is rewritten as  
\begin{subequations}
\label{eq:optimization_problem_wr}
\begin{equation}
\label{eq:objective_function_wr}    \min_{\mathbf{W},\mathbf{R}_{\mathbf{z}},\mathbf{U}} \mathrm{tr}\left(\mathbf{U}^{-1}\right),
\end{equation}
\begin{equation}
\label{eq:1}
    {\rm{s.t.}} \ \ \mathrm{tr}\left(\mathbf{W}+\mathbf{R}_{\mathbf{z}}\right)\leq {P}_{\text{max}},
\end{equation}
\begin{align}
    \label{eq:convex-qos}
    &\frac{1}{\gamma_{\text{min}}}\times\mathbf{h}_{\text{c}}^H\left(\mathbf{x}_{\text{p}}\right)\mathbf{G}\left(\mathbf{x}_{\text{p}}\right)\mathbf{W}\mathbf{G}^H\left(\mathbf{x}_{\text{p}}\right)\mathbf{h}_{\text{c}}\left(\mathbf{x}_{\text{p}}\right)\notag\\&\geq\mathbf{h}_{\text{c}}^H\left(\mathbf{x}_{\text{p}}\right)\mathbf{G}\left(\mathbf{x}_{\text{p}}\right)\mathbf{R}_{\mathbf{z}}\mathbf{G}^H\left(\mathbf{x}_{\text{p}}\right)\mathbf{h}_{\text{c}}\left(\mathbf{x}_{\text{p}}\right)+\sigma_{\text{c}}^2,
\end{align}
\begin{equation}
\label{eq:constraint_rank}
\mathbf{W}\succeq 0,
\end{equation}
\begin{equation}
    \label{eq:rank-one}\mathrm{rank}\left(\mathbf{W}\right)=1, 
\end{equation}
\begin{equation}
\eqref{eq:se_constraint0}, \eqref{eq:constraint_J},
\eqref{eq:se_constraint2},
\notag
\end{equation}
\end{subequations}
where $\gamma_{\text{min}}={2^{R_{\text{min}}}-1}$. Then, we adopt the penalty method to address the non-convex~\eqref{eq:rank-one}, which can be rewritten as $\Omega\triangleq\|\mathbf{W}\|_*-\|\mathbf{W}\|_2=0$,
where $\|\mathbf{W}\|_{*}$ and $\|\mathbf{W}\|_{2}$ denote the nuclear norm and spectral norm of the matrix $\mathbf{W}$. By putting $\Omega$ as a penalty term, \eqref{eq:optimization_problem_wr} can be reformulated as
\begin{subequations}
\label{eq:optimization_problem_wr_3}
\begin{equation}
\label{eq:objective_function_wr_3}    \min_{\mathbf{W},\mathbf{R}_{\mathbf{z}},\mathbf{U}} \mathrm{tr}\left(\mathbf{U}^{-1}\right)+\rho \Omega,
\end{equation}
\begin{equation}
{\rm{s.t.}} \ \
\eqref{eq:se_constraint0}, \eqref{eq:constraint_J},
\eqref{eq:se_constraint2}, \eqref{eq:1}-\eqref{eq:constraint_rank},
\end{equation}
\end{subequations}
where $\rho\geq0$ is the penalty factor. To further tackle the non-convexity of $\Omega$, we apply the SCA method that adopts the first-order Taylor expansion to obtain the convex upper bound:
\begin{equation}
    \label{penalty_term_sca}
\|\mathbf{W}\|_{*}-\|\mathbf{W}\|_{2}\leq \|\mathbf{W}\|_{*}-\bar{\mathbf{W}}^{\left(j\right)}\triangleq \Omega^{\left(j\right)},
\end{equation}
where $\bar{\mathbf{W}}^{\left(j\right)}\triangleq \|\mathbf{W}^{\left(j\right)}\|_{2}+\mathrm{tr}\left(\chi\left(\mathbf{W}^{\left(j\right)}\right)\chi\left(\mathbf{W}^{\left(j\right)}\right)^{H}\left(\mathbf{W}-\mathbf{W}^{\left(j\right)}\right)\right)$ with $\chi(\mathbf{W}^{\left(j\right)})$
representing the eigenvector of $\mathbf{W}^{\left(j\right)}$
associated with the largest eigenvalue. By substituting $\Omega^{(j)}$ into~\eqref{eq:objective_function_wr_3} in each SCA iteration, the surrogate problem can be solved by existing convex optimization tools such as CVX~\cite{cvx}. 
\subsection{Pinching Beamforming with Continuous Activation}
With given baseband beamforming, the optimization problem for the continuous PA activation is given by:
\begin{subequations}
\label{eq:optimization_problem_x}
\begin{equation}
\label{eq:objective_function_x}  \min_{\mathbf{x}_{\text{p}}} \mathrm{tr}\left(\mathrm{CRB}_{\boldsymbol{\psi}_{\text{s}}}\left(\mathbf{x}_{\text{p}}\right)\right),
\end{equation}
\begin{equation}
\label{eq:x_constraint_continuous}  
{\rm{s.t.}} \ \ \mathbf{x}_{\text{p}} \in \mathcal{P}_{\text{c}},
\end{equation}
\begin{equation}
 \eqref{eq:constraint_qos}.
 \notag
\end{equation}
\end{subequations}
Since problem~\eqref{eq:optimization_problem_x} is highly non-convex and exhibits multiple local optima with respect to $\mathbf{x}_{\text{p}}$, obtaining a high-quality solution using conventional gradient-based methods is challenging. To address this issue, we employ the PSO algorithm to efficiently explore the PAs positions space and search for a favorable solution through the collaborative evolution of multiple particles~\cite{Clercparticle2002}. Specifically, we first define a particle swarm consisting of $J$ particles, where each particle represents the positions of all PAs. The initial position of the $j$-th particle is denoted by 
\begin{equation}
\label{eq:initial position}
\mathbf{x}_{\text{p},j}^{\left(0\right)}=\left[\left(x_{\text{p},j}^1\right)^{\left(0\right)},\left(x_{\text{p},j}^2\right)^{\left(0\right)},\cdots, \left(x_{\text{p},j}^N\right)^{\left(0\right)}\right]^T.
\end{equation}
The initial velocity of the $j$-th particle is given by
\begin{equation}
\label{eq:initial position}
\mathbf{v}_{j}^{\left(0\right)}=\left[v_{j,1}^{\left(0\right)},v_{j,2}^{\left(0\right)},\cdots, v_{j,N}^{\left(0\right)}\right]^T.
\end{equation}
Subsequently, we set the initial local best position of the $j$-th particle as $\mathbf{x}_{\text{p},j}^*=\mathbf{x}_{\text{p},j}^{\left(0\right)}$. Then, the velocity and position of each particle in the $i$-th iteration are updated according to the following equations:
\begin{equation}
\begin{aligned}
\label{eq:pso-v-iter}
\mathbf{v}_{j}^{(i)}=&\omega\mathbf{v}_{j}^{(i-1)}+c_{1}\tau_{1}\left(\mathbf{x}_{\text{p},j}^*-\mathbf{x}_{\text{p},j}^{(i-1)}\right)\\&+c_{2}\tau_{2}\left(\mathbf{x}_{\text{p}}^*-\mathbf{x}_{\text{p},j}^{(i-1)}\right),
\end{aligned}
\end{equation}
\begin{equation}
\label{eq:pso-x-iter}
\mathbf{x}_{\text{p},j}^{\left(i\right)}[n]=B\left(\mathbf{x}_{\text{p},j}^{(i-1)}[n]+\mathbf{v}_j^{(i)}[n]\right),
\end{equation}
where $c_1$ and $c_2$ denote the cognitive and social learning factors, respectively, $\omega$ denotes the inertia weight, {$\mathbf{x}_{\text{p}}^*$ represents the global best position among all particles,} $\tau_1,\tau_2\sim\mathcal{U}[0,1]$ are random parameters, and $B(\cdot)$ represents the projection function, {defined as $B\left(x\right) = \max\left(0, \min\left(x, L_{\text{w}}\right)\right)$.} To achieve a balance between exploration and exploitation, the inertia weight $\omega$ is continuously decreased as follows:
\begin{equation}
\label{eq:inertia weight}
\omega=\omega_{\text{max}}-\frac{(\omega_{\text{max}}-\omega_{\text{min}})i}{I},
\end{equation}
where $\omega_{\text{min}}$ and $\omega_{\text{max}}$ denote the minimum and maximum value of $\omega$, respectively, and $I$ is the maximum number of iterations. To tackle constraints~\eqref{eq:constraint_qos} and~\eqref{eq:x_constraint_continuous}, we define a fitness function by putting the corresponding penalty terms into the original objective function~\eqref{eq:objective_function_x} as follows:
\begin{equation}
\label{eq:fitness function}
\begin{aligned}F\left(\mathbf{x}_{\text{p},j}\right)&=\mathrm{tr}\left(\mathrm{CRB}_{\boldsymbol{\psi}_{\text{s}}}\left(\mathbf{x}_{\text{p},j}\right)\right)\\&+\mu_{1}\bar{R}\left(\mathbf{x}_{\text{p},j}\right)+\mu_{2}D\left(\mathbf{x}_{\text{p},j}\right),\end{aligned}
\end{equation}
where $\mu_1$ and $\mu_2$ are the positive penalty parameters. $\bar{R}\left(\mathbf{x}_{\text{p},j}\right)$ and $D\left(\mathbf{x}_{\text{p},j}\right)$ are penalty terms that evaluate the violation of the constraint~\eqref{eq:constraint_qos} and~\eqref{eq:x_constraint_continuous}, which are defined as follows:
\begin{equation}
    \label{eq:penalty_2}
\bar{R}\left(\mathbf{x}_{\text{p},j}\right)=\max\left\{R_{\text{min}}-R_{\text{c}}\left(\mathbf{x}_{\text{p},j}\right),0\right\},
\end{equation}
\begin{equation}
\begin{aligned}
\label{eq:penalty_1}
D\left(\mathbf{x}_{\text{p},j}\right)&=\sum_{n=1}^{N}\sum_{n'=n+1}^N\max\left\{\Delta-\left|x_{\text{p},j}^{n}-x_{\text{p},j}^{n'}\right|,0\right\}.
\end{aligned}
\end{equation}
By setting the penalty factors $\mu_1$ and $\mu_2$ as sufficiently large values, the constraints~\eqref{eq:constraint_qos} and~\eqref{eq:x_constraint_continuous} can be satisfied by driving the particle to the feasible region. For each particle, its local and global best positions are updated by the current position when the obtained fitness value gets smaller. The details of the PSO algorithm are given in $\textbf{Algorithm~\ref{alg:pso}}$.
\begin{algorithm}[t]
\caption{PSO Algorithm for Solving Problem~\eqref{eq:optimization_problem_x}}
\label{alg:pso}
\begin{algorithmic}[1]
%\STATE \textbf{Input:} 
%\STATE \textbf{Output:} $\tilde{\mathbf{x}^{\text{p}}}$.
\STATE Initialize $\mathbf{x}_{\text{p},j}^{\left(0\right)}$ and $\mathbf{v}_{j}^{\left(0\right)}$ for each particle.
\STATE Calculate the fitness value for each particle.
\STATE Set the local best position $\mathbf{x}_{\text{p},j}^* = \mathbf{x}_{\text{p},j}^{\left(0\right)}$ and the global best position $\mathbf{x}_{\text{p}}^* = \mathop{\arg\min} \left\{F\left(\mathbf{x}_{\text{p},1}^{\left(0\right)}\right), \dots, F\left(\mathbf{x}_{\text{p},J}^{\left(0\right)}\right)\right\}$.
\FOR{$i = 1$ to $I$}
    \STATE Update the inertia weight $\omega$ in~\eqref{eq:inertia weight}.
    \FOR{$j = 1$ to $J$}
        \STATE Calculate the velocity and position of the $j$-th particle with~\eqref{eq:pso-v-iter} and~\eqref{eq:pso-x-iter}, respectively.
        \STATE Calculate the fitness value of the $j$-th particle with~\eqref{eq:fitness function}.
        \IF{$F\left(\mathbf{x}_{\text{p},j}^{\left(i\right)}\right) < F\left(\mathbf{x}_{\text{p},j}^*\right)$}
            \STATE Update $\mathbf{x}_{\text{p},j}^* = \mathbf{x}_{\text{p},j}^{\left(i\right)}$.
        \ENDIF
        \IF{$F\left(\mathbf{x}_{\text{p},j}^{\left(i\right)}\right) < F\left(\mathbf{x}_{\text{p}}^*\right)$}
            \STATE Update $\mathbf{x}_{\text{p}}^*=\mathbf{x}_{\text{p},j}^{(i)}$.
        \ENDIF
    \ENDFOR
\ENDFOR
\RETURN $\mathbf{x}_{\text{p}}^*$.
\end{algorithmic}
\end{algorithm}

\subsection{Pinching Beamforming with Discrete Activation} 
For the discrete activation case, the $\mathbf{x}_{\text{p}}$ optimization problem can be formulated as 
\begin{subequations}
\label{eq:optimization_problem_x_d}
\begin{equation}
\label{eq:objective_function_x_d}  \min_{\mathbf{x}_{\text{p}}} \mathrm{tr}\left(\mathrm{CRB}_{\boldsymbol{\psi}_{\text{s}}}\left(\mathbf{x}_{\text{p}}\right)\right),
\end{equation}
\begin{equation}
\label{eq:x_constraints_d_1}  
{\rm{s.t.}}\ \ \mathbf{x}_{\text{p}}\in\mathcal{P}_{\text{d}},
\end{equation}
\begin{equation}
\label{eq:x_constraints}  
  \eqref{eq:constraint_qos}.
 \notag
\end{equation}
\end{subequations}
For solving the integer programming problem~\eqref{eq:optimization_problem_x_d}, we invoke the matching theory~\cite{matching1} to explore the low-complexity and near-optimal solution. Specifically, under the constraint~\eqref{eq:x_constraints_d_1}, the relationship between the PA set $\mathcal{N}$ and the candidate position set $\mathcal{P}$ can be regarded as a \textit{two-sided one-to-one} matching problem.
\begin{definition}
In the proposed two-sided one-to-one matching model, $\varphi$ denotes a mapping $\left(\mathcal{N}\cup\mathcal{P}\right) \rightarrow \left(\mathcal{N}\cup\mathcal{P}\right)$ , satisfying
\begin{enumerate}
    \item $\varphi(n)\in \mathcal{P}, |\varphi(n)|=1, \forall n \in \mathcal{N}$;
    \item $\varphi(p_a)\in \mathcal{N}, |\varphi(p_n)|=1, \forall p_a \in \mathcal{P}$;
    \item $\varphi(n)=p_a \Leftrightarrow \varphi(p_a)=n$.
\end{enumerate}
\end{definition}
\noindent

Since our objective is to minimize the trace of CRB by properly matching PAs with the candidate positions, all matching participants share the same global utility
$U\left(\varphi\right)=-\mathrm{tr}\left(\mathrm{CRB}_{\boldsymbol{\psi}_{\text{s}}}\left(\varphi\right)\right)$,
where $\mathrm{CRB}_{\boldsymbol{\psi}_{\text{s}}}\left(\varphi\right)$ is the achievable CRB under matching state $\varphi$. Accordingly, a common ``preference" relation over
the feasible matching states is induced by the global utility, i.e.,
\begin{equation}
    \varphi' \succ \varphi
    \quad\Longleftrightarrow\quad
    U(\varphi')>U(\varphi).
\end{equation} 
The pair $(n,n')$ is referred
to as a \textit{swap-blocking pair} under matching $\varphi$ if
$U\left(\varphi_{n\leftrightarrow n'}\right)>U(\varphi)$, where $\varphi_{n\leftrightarrow n'}=\left\{\varphi\setminus\left\{(n,\varphi(n)),(n',\varphi(n'))\right\}\right\}\cup \left\{(n,\varphi(n')),(n',\varphi(n))\right\}$. Note that, to incorporate vacant candidate positions into the swap process, we introduce dummy PAs, denoted by the set $\mathcal{N}_{\text{d}}\triangleq\left\{d_1,\dots,d_B\right\}$,  such that each unoccupied position is temporarily matched with a dummy PA. Dummy PAs are assumed
to have no physical functionality and make no contribution to the
system utility. Suppose that PA $n$ is currently matched with position
$p_a$, whereas position $p_{a'}$ is occupied by a dummy PA $d_b$,
The relocation of PA $n$ from $p_a$ to $p_{a'}$ is then represented
by the swap matching $\varphi_{n\leftrightarrow d_b}$. A matching is swap-stable if it admits no swap-blocking pair. 

To obtain the swap-stable matching, we develop a swap-based matching algorithm as shown in \textbf{Algorithm~\ref{alg:matching}}. For the initialization, PAs are randomly matched to candidate positions on the condition of satisfying the minimum communication rate. Subsequently, each PA searches for all other PAs (including dummy PAs) to check whether there exists swap-blocking pairs. Since a feasible swap is accepted only if it strictly
increases the common utility, the algorithm
cannot revisit any previously visited matching state, which guarantees that \textbf{Algorithm~\ref{alg:matching}} converges to a swap-stable matching after a finite number of swaps. 
\begin{algorithm}[!tp]
\caption{Matching Algorithm for Solving Problem~\eqref{eq:optimization_problem_x_d}}
\label{alg:matching}
\begin{algorithmic}[1]
\STATE Randomly match each PA $n$ with a candidate position in $\mathcal{P}$ on the condition of satisfying constraint \eqref{eq:constraint_qos}.
\STATE Match each unoccupied position with $d_b\in\mathcal{N}_{\text{d}}$.
    \REPEAT
    \FOR{$n \in \mathcal{N}$}
    \FOR{$n'\in \left\{\mathcal{N}\setminus \left\{n\right\}\right\}\cup\mathcal{N}_d$}
        \IF{$\left(n,n'\right)$ is a blocking pair and \eqref{eq:constraint_qos} is satisfied under $\varphi_{n\leftrightarrow n'}$}
            \STATE Update the matching state $\varphi = \varphi_{n\leftrightarrow n'}$.
        \ENDIF
        \ENDFOR
    \ENDFOR
    \UNTIL{$\nexists (n,n')$ blocks $\varphi$.}
\RETURN $\varphi$.
\end{algorithmic}
\end{algorithm}
\vspace{-0.3cm}
\subsection{Overall Algorithm}
{The AO-based algorithm for solving the original problem~\eqref{eq:optimization_problem} is given in~\textbf{Algorithm~\ref{alg:AO}}, where the baseband and pinching beamforming subproblems are solved iteratively until the decrement of CRB gets below the predefined threshold. Since the trace of CRB is lower-bounded and non-increasing for optimizing $\left\{\mathbf{W},\mathbf{R}_{\mathbf{z}}\right\}$ and $\mathbf{x}_{\text{p}}$,  \textbf{Algorithm~\ref{alg:AO}} is guaranteed to converge within limited number of iterations. The computational complexity of solving the beamforming subproblem is given by $O(I_{\text{out}}I_{\text{in}} 2^{3.5})$, where $I_{\text{out}}$ and $I_{\text{in}}$ denote the number of iterations required for updating the penalty parameter and the SCA convergence, respectively. For the continuous activation case, the computational complexity of the PSO-based pinching beamforming algorithm is given by $O(IJN)$. For the discrete activation case, the matching-based pinching beamforming algorithm yields the complexity of $O(I_{\text{match}}NA)$, where $I_{\text{match}}$ is the number of the outermost cycle in the matching algorithm.}
\begin{algorithm}[tp]
    \renewcommand{\algorithmicrequire}{\textbf{Input:}}
    \renewcommand{\algorithmicensure}{\textbf{Output:}}
  \caption{Proposed AO-based algorithm for solving problem $\eqref{eq:optimization_problem}$}
  \begin{algorithmic}[1]
     \label{alg:AO}
      \STATE  Initialize  $\mathbf{x}_{\text{p}}^{\left(0\right)}$, $\mathbf{w}^{\left(0\right)}$, and $\mathbf{R}_{\mathbf{z}}^{\left(0\right)}$.
      \STATE Set iteration index $q = 0$.
      
      \REPEAT
         \STATE Calculate $\mathbf{w}^{\left(q+1\right)}$ and $\mathbf{R}_{\mathbf{z}}^{\left(q+1\right)}$ with given $\mathbf{x}_{\text{p}}^{(q)}$.
         \STATE Calculate $\mathbf{x}_{\text{p}}^{(q+1)}$ with given $\mathbf{w}^{(q+1)}$ and $\mathbf{R}_{\mathbf{z}}^{(q+1)}$.   
      \STATE $q =q+1$.
      \UNTIL the decrement of $\mathrm{CRB}$ is below the threshold $\epsilon$.
      \RETURN $\mathbf{x}_{\text{p}}, \mathbf{w}, \mathbf{R}_{\mathbf{z}}$.
  \end{algorithmic}
\end{algorithm}

\section{Simulation Results}
In this section, numerical results are provided to demonstrate the benefits of multi-mode PASS in enhancing ISAC performance and to evaluate the effectiveness of the proposed algorithms.
\subsection{Simulation Setup}
The BS is equipped with a rectangular dielectric waveguide with the length $L_{\text{w}} = 10$~m deployed at the height of $d_{\text{w}} = 3$~m. An equal number of PAs is allocated to the two considered modes, i.e., $\left|\mathcal{N}_1\right|=\left|\mathcal{N}_2\right|=\frac{N}{2}$. 
The minimum spacing $\Delta$ among PAs is set to $\frac{\lambda}{2}$, with $\lambda$ denoting the free-space wavelength at $28$~GHz. The waveguide and PAs parameters setup are consistent with the example given in Fig.~\ref{fig:hfss-simulation}. All PAs are with the longitudinal length of $6.545$~mm.
%The length and width of the rectangular cross-section of the waveguide are set to $12$~mm and $1.5$~mm. The refractive index of the cladding medium is $n_{\text{clad}}=1$, and the relative permittivity of the waveguide core is $\epsilon_r=4$ with $n_{\text{core}}=2$. At the carrier frequency $f_{\text{c}}=28$~GHz, the propagation constants for the two orthogonal guided modes in the waveguide, i.e., the quasi-$\mathrm{TE}_{00}$ and quasi-$\mathrm{TE}_{10}$ modes, are $\beta_1 = 688.3$~rad/m and $\beta_2 = 602.8$~rad/m, respectively.
The ULA consists of $K = 4$ antenna elements, where the antenna spacing is set as $\frac{\lambda}{2}$. The CU and the ST are randomly distributed within a rectangular area $\mathcal{S}=\{\left[x,y,0\right],x\in(0,L_{\text{w}}),y\in(-10,10)~\text{m}\}$. The reference channel amplitude gain at a distance of $1$~m is $\frac{\lambda}{4\pi}$. The number of time-domain snapshots is $T = 256$. The maximum BS transmit power is $P_{\text{max}}=20$~dBm, and the noise powers at the CU and ULA are set as $\sigma_{\text{c}}^2 = \sigma_{\text{s}}^2 = -90$~dBm. The minimum data rate for the CU is $R_{\text{min}} = 2$~bps/Hz. For the continuous PAs activation case, the penalty-based PSO algorithm exploits $J=200$ particles with the maximum number of iterations $I=100$. The cognitive and social learning coefficients are set as $c_1=1.8$ and $c_2=0.8$, respectively. The minimum and the maximum inertia weights are set as $\omega_{\text{min}}=0.1$ and $\omega_{\text{max}}=0.9$, respectively. For the discrete PAs activation case, we set the number of available PAs positions as $A=50$. Unless otherwise specified, the following simulations adopt the same parameter settings as those described above. All simulation results are averaged over 50 independent realizations. The following benchmark schemes are considered for performance comparison.

\begin{itemize}

\item \textbf{Single-mode PASS, continuous:}
In this case, only one guided mode is excited in the waveguide, indicating that only one data stream can be transmitted. As such, the BS utilizes the information signal for sensing, without introducing an additional dedicated sensing signal. The resultant pinching beamforming is optimized under the continuous activation assumption, by applying~\textbf{Algorithm~\ref{alg:pso}}.

\item \textbf{Conventional MIMO:}
In this case, the BS adopts a ULA for signal transmission. For a fair comparison with the multi-mode PASS, the ULA adopts the sub-connected hybrid beamforming structure, where $N$ antennas are equally partitioned into $2$ subarrays, each connected to one RF chain. The ULA is centered at $\left(0,0,3\right)$~m, and its antenna elements are deployed along the $y$-axis with a uniform spacing of $\frac{\lambda}{2}$. 
The hybrid beamforming is optimized with the algorithms proposed in ~\cite{Massive-mimo}.
\end{itemize}
 
\subsection{RCRB Versus the BS Transmit Power}
\begin{figure}
    \centering    \includegraphics[width=0.85\linewidth]{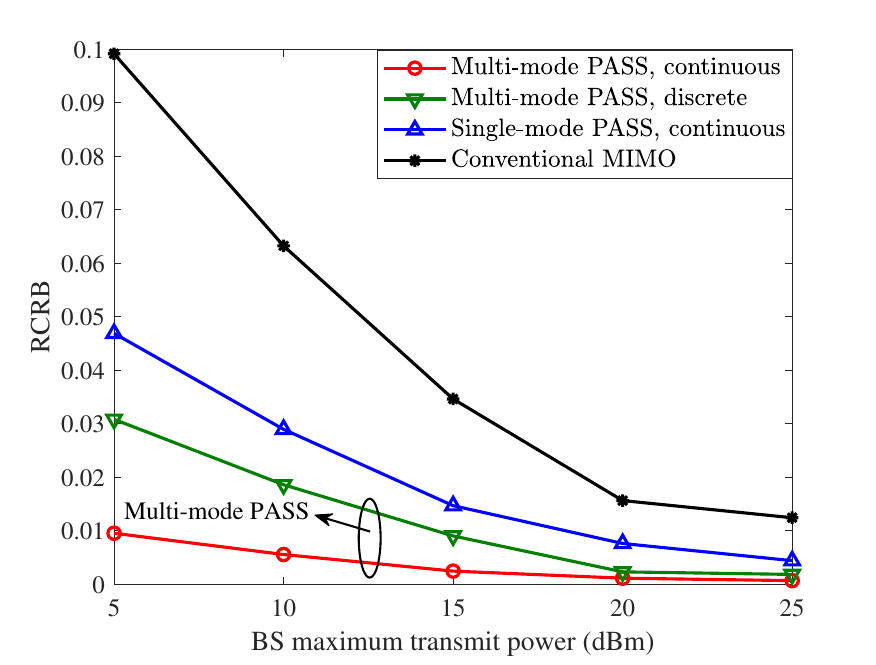}
    \caption{RCRB versus BS maximum transmit power $P_{\text{max}}$, with $N=6$.}
    \label{fig:RCRB-P}
\end{figure}
Fig.~\ref{fig:RCRB-P} depicts the root CRB (RCRB) versus maximum BS transmit power $P_{\text{max}}$. As observed, the multi-mode PASS scheme under the continuous activation case achieves the lowest RCRB among all considered schemes. This is because, on the one hand, PASS can significantly reduce the path loss by flexibly moving PAs towards the CU and the ST. On the other hand, multiple guided modes within the waveguide can provide additional modal DoFs to enhance sensing and communication performance simultaneously. The discrete activation case exhibits slightly degraded performance due to the limited deployment flexibility caused by position discretization. It is also observed that, the single-mode PASS exhibits a significant performance degradation compared to the proposed multi-mode PASS schemes. This is expected, as the single-mode PASS solely utilizes the communication signal waveform for sensing the target, resulting in a rank-one transmit covariance matrix and limited spatial DoFs for jointly shaping the communication and sensing functionalities. This performance degradation get more pronounced under smaller $P_{\text{max}}$, where the single-mode PASS have to consume most of the power to guarantee the minimum data rate constraint of the CU. Furthermore, the conventional MIMO scheme achieves the highest RCRB. Due to the fixed-position antenna array, the MIMO structure lacks the flexible channel reconfiguration capability, which heavily restricts the spatial DoFs compared to the PASS architecture.

\subsection{RCRB Versus the Number of PAs}
\begin{figure}
        \centering    \includegraphics[width=0.85\linewidth]{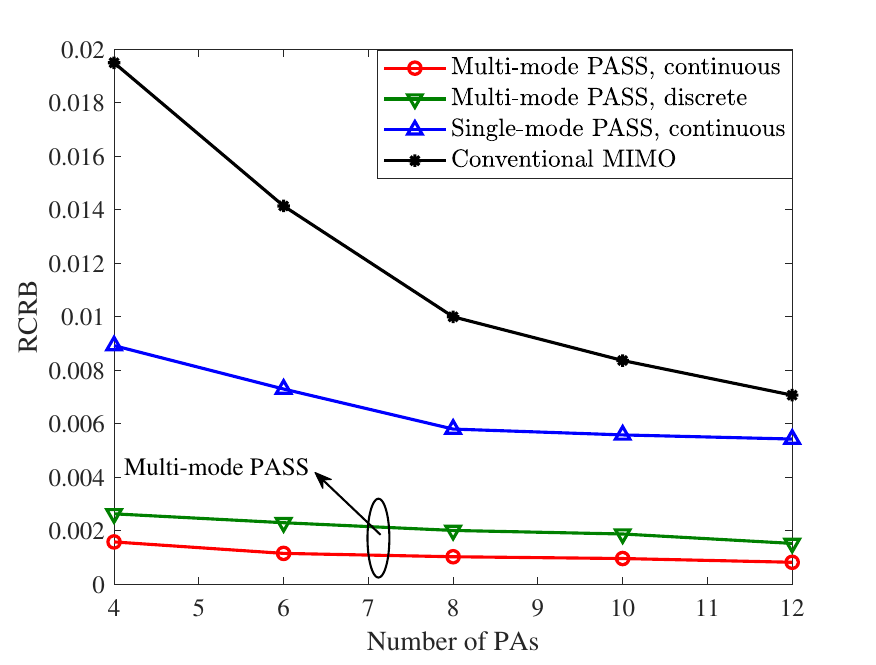}
    \caption{RCRB versus number of PAs $N$.}
    \label{fig:RCRB-N}
\end{figure}
Fig.~\ref{fig:RCRB-N} illustrates the RCRB versus the number of PAs. It can be observed that the RCRB decreases with the increment of $N$ for all considered schemes. For the PASS schemes, additional PAs provide more spatial DoFs for pinching beamforming. For the conventional MIMO scheme, although the sensing performance improves owing to the increased array gain, it remains inferior to that of the PASS schemes due to the limited channel reconfiguration capability. Notably, the RCRB reduction of the single-mode PASS becomes less pronounced as $N$ increases. This is because, although additional PAs provide more spatial DoFs for pinching beamforming, the single-mode architecture still suffer from the restricted multiplexing gain.  Consequently, given that the minimum data-rate requirement of the CU needs to be satisfied, further increasing the number of PAs yields only marginal sensing gains.

\subsection{RCRB Versus the Waveguide Length}
\begin{figure}
        \centering    \includegraphics[width=0.85\linewidth]{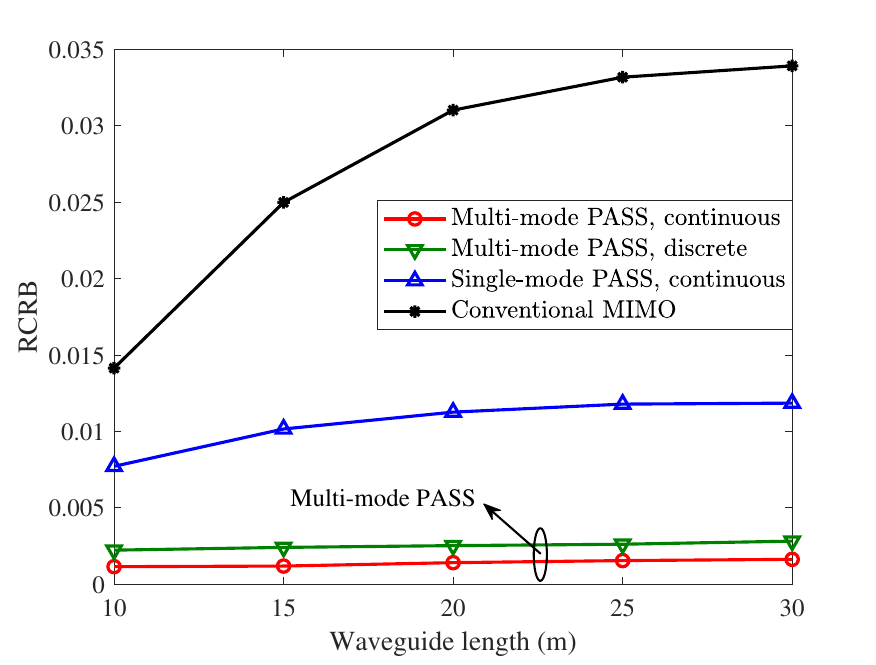}
    \caption{RCRB versus waveguide length $L_{\text{w}}$, with $N=6$.}
    \label{fig:RCRB-L}
\end{figure}
{In Fig.~\ref{fig:RCRB-L}, we evaluate the impact of the waveguide length $L_{\text{w}}$ on the ISAC performance in terms of the RCRB.} It is first observed that, the conventional MIMO scheme exhibits a significantly worse RCRB with the increment of $L_{\text{w}}$. This is because, a larger $L$ extends the CU and ST distribution area, leading to higher probability that the CU and ST farther away from the antenna array, which can intensify the path loss. Moreover, the RCRB achieved by the single-mode PASS increases with the increment of $L$. This is because, as the CU and ST become more spatially separated, the rank-one transmit covariance makes it difficult to simultaneously establish favorable beamforming gains toward the CU and ST. Consequently, more spatial resources must be devoted to satisfying the communication rate requirement, which further compromises the sensing performance and leads to a larger CRB. In contrast, the proposed multi-mode PASS schemes maintain an extremely low RCRB level, which demonstrates the effectiveness of multi-mode PASS for accommodating both the communication and sensing performances.

\subsection{RCRB Versus the Minimum Communication Rate}
\begin{figure}
        \centering    \includegraphics[width=0.85\linewidth]{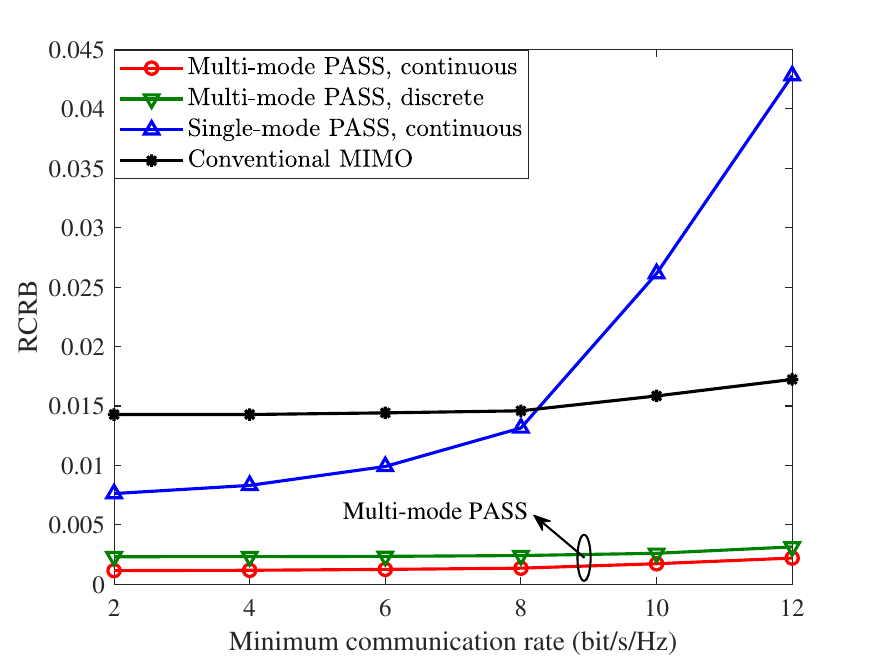}
    \caption{RCRB versus minimum communication rate $R_{\text{min}}$, with $N=6$.}
    \label{fig:RCRB-tradeoff}
\end{figure}
In Fig.~\ref{fig:RCRB-tradeoff}, we investigate the ISAC performance under different minimum communication rate requirements, i.e., $R_{\text{min}}$. It can be observed that, the sensing RCRBs for all schemes exhibit a monotonically increasing trend as $R_{\text{min}}$ increases. This is expected, as a larger $R_{\text{min}}$ forces the BS to allocate more spatial resources toward the CU to satisfy the minimum data rate constraint, which leads to a higher sensing RCRB. Meanwhile, the proposed multi-mode PASS schemes consistently outperform the benchmark schemes and maintain low RCRB values over the range of $R_{\text{min}}$. This is attributed to the fact that the multi-mode transmission provides additional modal multiplexing gains, allowing the system to satisfy the increasingly stringent minimum data requirement while preserving favorable sensing performance. In contrast, the sensing performance of the single-mode PASS deteriorates significantly as $R_{\text{min}}$ increases. The reason is that, the pinching beamformer should provide a higher beamforming gain toward the CU with the increment of $R_{\text{min}}$, leaving less flexibility to shape the transmitted signal toward the ST. Consequently, the communication-sensing beamforming tradeoff becomes more pronounced.
It is noted that the RCRB achieved by the single-mode PASS scheme even exceeds that of the conventional MIMO scheme when $R_{\text{min}}$ is sufficiently large.

\subsection{RCRB Versus the Number of PAs Candidate Positions}
\begin{figure}
        \centering    \includegraphics[width=0.85\linewidth]{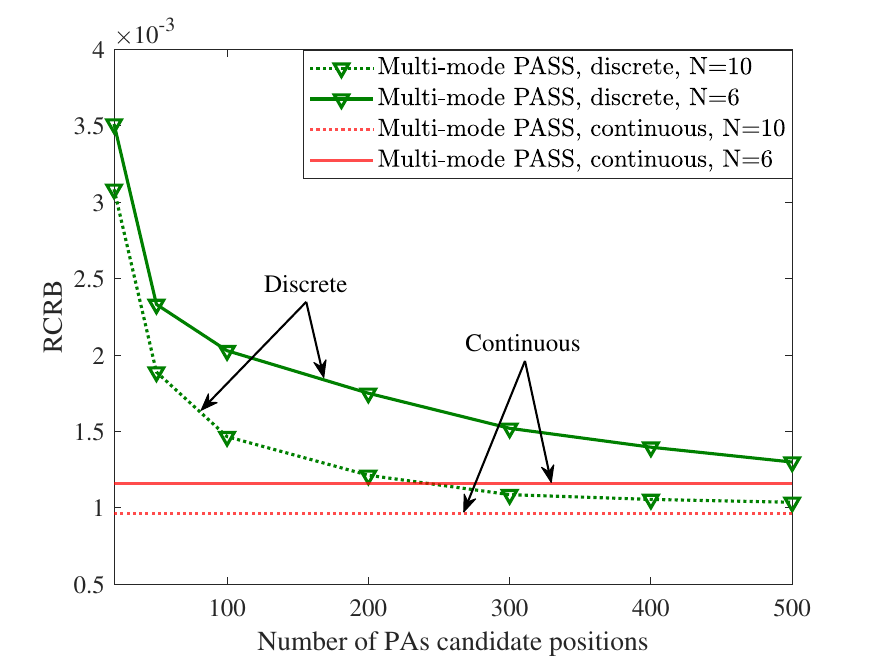}
    \caption{RCRB versus number of PAs candidate positions $A$.}
    \label{fig:RCRB-A}
\end{figure}

{In Fig.~\ref{fig:RCRB-A}, we investigate the impact of the number of PAs candidate positions on the ISAC performance. As can be observed, the achieved RCRB decreases with the increment of the number of candidate positions in the discrete activation case. The reason is that, a larger number of candidate positions provides enhanced precisions for PAs activation. It is worth noting that, when $A$ increases to sufficiently large values, the performance gain of the continuous case over the discrete case is marginal. For example, when there are $500$ candidate positions over the $10$~m waveguide, there is only a $9.7\%$ and $6.8\%$ increment in RCRB caused by the discretization-induced loss, for $N=6$ and $N=10$, respectively. This is because, a dense set of candidate positions provides the discrete scheme with sufficient flexibility to activate PAs. It is also observed that increasing the number of PAs from $N=6$ to $N=10$ achieves lower converged RCRB values, as a larger number of PAs provides higher spatial DoFs for pinching beamforming. }

\subsection{Proposed AO Algorithm Convergence Behavior}
\begin{figure}
    \centering    \includegraphics[width=0.85\linewidth]{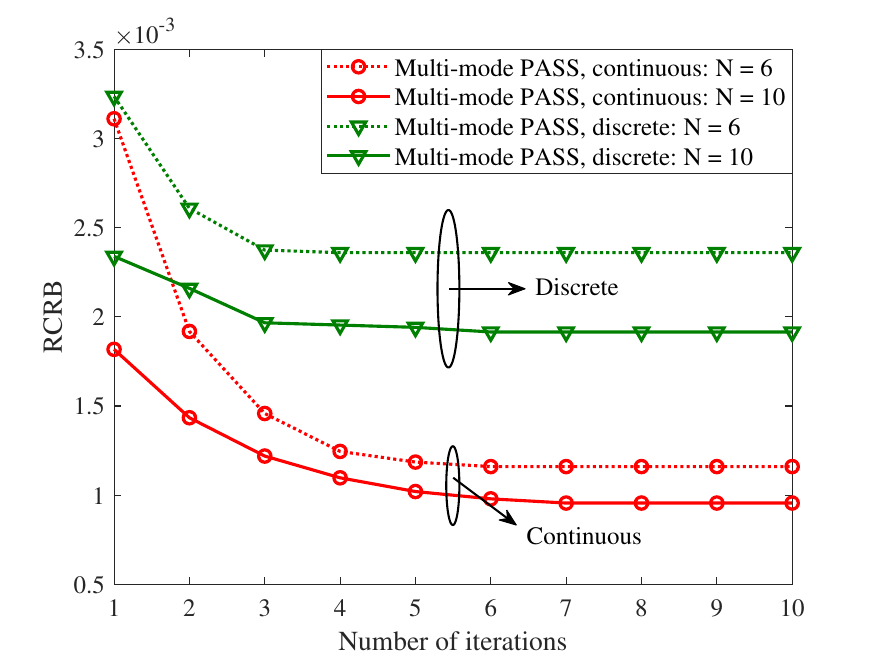}
    \caption{Convergence performance of {Algorithm~\ref{alg:AO}}.}
    \label{fig:RCRB-iter}
\end{figure}
Fig.~\ref{fig:RCRB-iter} illustrates the convergence behavior of the proposed AO algorithm under different numbers of PAs for both continuous and discrete PAs activation cases, where the RCRB is plotted against the number of iterations. As observed, the RCRBs decrease rapidly at the beginning of the iterations and gradually converge to stable values for all considered cases, demonstrating the effectiveness of the proposed AO-based algorithm. Meanwhile, we can find that a larger number of PAs leads to more iterations for convergence. However, even when $N = 10$, the proposed AO-based algorithm still converges within only a few iterations. This fast convergence speed for both continuous and discrete activation cases implies the possible application of the proposed AO-based algorithm in practice.

\section{Conclusions}
In this paper, we first proposed the physical model of multi-mode PASS, where multiple guided modes were excited within a single waveguide to act as independent signal-bearing channels. We explicitly investigated the undesired guided-modes coupling under the local PA-induced electromagnetic perturbations, and revealed the resulting IMI issue. To construct a tractable signal model, we established the hardware-design conditions for achieving the IMI-free regime, which was further verified through full-wave simulations. We then studied ISAC as a representative application scenario of multi-mode PASS, and formulated a joint baseband and pinching beamforming problem. To address the resulting non-convex optimization problems, an AO-based framework was developed. For the baseband beamforming, a penalty-based SCA algorithm was invoked. For the pinching beamforming, the PSO algorithm and two-sided matching algorithm were proposed for the continuous and discrete activation cases, respectively. Numerical results obtained in the ISAC scenario demonstrated the advantages of the proposed IMI-free multi-mode PASS model over the single-mode PASS in terms of providing additional modal DoFs.
%, through validating the effectiveness of the additional modal DoFs in alleviating the communication-sensing tradeoff. 
\begin{appendices}
\label{app:IMI-free}
\section{Proof of~\textbf{Proposition~1}}
\label{app:proof}
The IMI observed at mode $m$ contains the direct leakage from other guided modes $m'\neq m$, as well as the corresponding PA-mediated and non-PA-mediated higher-order transition components. 
First consider the direct signal leakage from guided mode $m'$ to $m$. Evaluating the integral in Eq.~\eqref{eq:m-n-direct-transition} yields
\begin{equation}
    T_{mm'}^{(1)}(x)
    =
    -j\kappa_{mm'}x
    e^{-j\frac{(\beta_{m'}-\beta_m)x}{2}}
    \operatorname{sinc}
    \left(
        \frac{(\beta_{m'}-\beta_m)x}{2}
    \right).
    \label{eq:first_order_wg_transfer_sinc}
\end{equation}
Based on
\begin{equation}
    \left|
        x\operatorname{sinc}
        \left(
            \frac{(\Delta \beta_{mm'}) x}{2}
        \right)
    \right|
    =
    \left|
        \frac{
            2\sin(\Delta\beta_{mm'} x/2)
        }{
            \Delta\beta_{mm'}
        }
    \right|
    \leq
    \frac{2}{|\Delta\beta_{mm'}|},
    \label{eq:sinc_upper_bound}
\end{equation}
where $\Delta\beta_{mm'}=\beta_{m'}-\beta_m$, we obtain
\begin{equation}
    \left|
        T_{mm'}^{(1)}(x)
    \right|
    \leq
    \frac{
        2|\kappa_{mm'}|
    }{
        |\beta_{m'}-\beta_m|
    }
    =
    \epsilon_{mm'}
    \leq
    \epsilon.
    \label{eq:first_order_wg_bound}
\end{equation}
Hence, the direct signal leakage from $m'$ to $m$ is
bounded by $\mathcal{O}(\epsilon)$.

Given the modal-selectivity characteristic of each PA, the PA-mediated transition from guided mode $m'$ to $m$ contains at least one undesired transition. Assuming that the considered modal couplings are reciprocal, the direct transition between the PA and its unmatched guided mode $m''$ satisfies
\begin{equation}
    \left|
        T_{m''\text{p}}^{(1)}(x)
    \right|=\left|
        T_{\text{p}m''}^{(1)}(x)
    \right|
    \leq
    \frac{
        2|\kappa_{\text{p}m''}|
    }{
        |\beta_{\text{p}}-\beta_{m''}|
    }
    =
    \epsilon_{\text{p}m''}
    \leq
    \epsilon.
    \label{eq:first_order_pa_bound}
\end{equation}

Under the higher-order non-resonance assumption, any higher-order coupling path containing at least one undesired modal transition is assumed not to experience appreciable resonant enhancement~\cite{Huang2009CMT}. Therefore, a coupling path containing $o$ undesired transitions preserves the amplitude transfer coefficient of order $\mathcal{O}(\epsilon^o)$. Since the coupling path from $m'$ to $m$ involves at least one undesired modal transition, we have
\begin{equation}
    T_{mm'}(x)
    =
    \mathcal{O}(\epsilon),
    \qquad m\neq m'.
    \label{eq:total_cross_transfer_order}
\end{equation}
Substituting \eqref{eq:total_cross_transfer_order} into
\eqref{eq:modal_input_output_relation} gives
\begin{equation}
    A_m(x)
    =
    T_{mm}(x)s_m
    +
    \mathcal{O}(\epsilon)
    \sum_{m'\neq m}s_{m'}.
    \label{eq:modal_signal_approximation}
\end{equation}

Assume that the transmitted symbols are mutually independent and zero mean, the aggregated IMI power associated with mode $m$ is given by
\begin{align}
    P_{m}^{\mathrm{I}}
    =
    \sum_{m'\neq m}
    |T_{mm'}(x)|^2P_{\text{s}}=
    \mathcal{O}
    \left(
        (M-1)\epsilon^2
    \right)P_{\text{s}}.
    \label{eq:interference_power_proof}
\end{align}
Therefore, the IMI-to-signal power ratio satisfies
\begin{equation}
    \frac{P_{m}^{\mathrm{I}}}{P_{\text{s}}}=\mathcal{O}\left(\left(M-1\right)\epsilon^2\right), \forall m.
\end{equation}

\section{Derivations of $\mathbf{J}_{11}$, $\mathbf{J}_{12}$, and $\mathbf{J}_{22}$}
\label{app:FIM}
According to~\cite{CRB}, the received echo signal matrix $\mathbf{Y}_{\text{s}}\left(\mathbf{x}_{\text{p}},\mathbf{w},\mathbf{R}_{\mathbf{z}}\right)$ can be vectorized as follows:
\begin{equation}
\label{eq:vec-y}
\mathbf{y}_{\text{s}}=\varrho\text{vec}\left(\mathbf{a}\left(\theta_{\text{b}}\right)\mathbf{h}_{\text{s}}^H\left(\mathbf{x}_{\text{p}}\right)\mathbf{G}\left(\mathbf{x}_{\text{p}}\right)\mathbf{S}\right)+\text{vec}\left(\mathbf{N}_{\text{b}}\right).
\end{equation}
$\mathbf{y}_{\text{s}}$ follows the Gaussian distribution, i.e., 
$\mathbf{y}_{\text{s}}\sim\mathcal{CN}\left(\mathbf{q},\sigma_{\text{s}}^2\mathbf{I}_{KT}\right)$, where $\mathbf{q}=\varrho\text{vec}\left(\mathbf{a}\left(\theta_{\text{b}}\right)\mathbf{h}_{\text{s}}^H\left(\mathbf{x}_{\text{p}}\right)\mathbf{G}\left(\mathbf{x}_{\text{p}}\right)\mathbf{S}\right)$. 
We define the unknown parameter $\boldsymbol{\xi}=[{\boldsymbol{\psi}}^T,{\boldsymbol{\varrho}}^T]$, where $\boldsymbol{\psi}=[x_{\text{s}},y_{\text{s}}]^T$, $\boldsymbol{\varrho}=[\varrho_{\text{r}},\varrho_{\text{i}}]^T$, and $\varrho_{\text{r}}$ and $\varrho_{\text{i}}$ denote the real and imaginary parts of $\varrho$, respectively. Then, the FIM is given by~\cite{CRB}
\begin{equation}
\label{eq:element of FIM}
\mathbf{J}_{\boldsymbol{\xi}}=\frac{2}{\sigma_{\text{s}}^2}\text{Re}\left\{\frac{\partial\mathbf{q}^H}{\partial\boldsymbol{\xi}}\frac{\partial\mathbf{q}}{\partial\boldsymbol{\xi}}\right\}=\begin{bmatrix}
\ \mathbf{J}_{11},\mathbf{J}_{12}\\\mathbf{J}_{12}^T,\mathbf{J}_{22}
\end{bmatrix},
\end{equation}
where 
$\mathbf{J}_{11}=\begin{bmatrix}J_{x_{\text{s}}x_{\text{s}}} J_{x_{\text{s}}y_{\text{s}}}\\J_{x_{\text{s}}y_{\text{s}}} J_{y_{\text{s}}y_{\text{s}}}
\end{bmatrix}$, $\mathbf{J}_{12}=\begin{bmatrix}J_{x_{\text{s}}\varrho_{\text{r}}} J_{x_{\text{s}}\varrho_{\text{i}}}\\J_{y_{\text{s}}\varrho_{\text{r}}} J_{y_{\text{s}}\varrho_{\text{i}}}
\end{bmatrix}$, and $\mathbf{J}_{22}=\begin{bmatrix}J_{\varrho_{\text{r}}\varrho_{\text{r}}} & 0\\0 &J_{\varrho_{\text{i}}\varrho_{\text{i}}}
\end{bmatrix}$. Subsequently, defining $\mathbf{D}\triangleq\mathbf{a}\left(\theta_{\text{b}}\right)\mathbf{h}_{\text{s}}^H(\mathbf{x}_{\text{p}})$, the partial derivatives of $\mathbf{q}$ with respect to $\boldsymbol{\psi}$ and $\boldsymbol{\beta}$ are given by
\begin{equation}
\label{eq:derivation of q_1}  
\frac{\partial\mathbf{q}
}{\partial \boldsymbol{\psi}^T}=\varrho\left[\text{vec}\left(\dot{\mathbf{D}}_{x_{\text{s}}}\mathbf{G}\left(\mathbf{x}_{\text{p}}\right)\mathbf{S}\right),\text{vec}\left(\dot{\mathbf{D}}_{y_{\text{s}}}\mathbf{G}\left(\mathbf{x}_{\text{p}}\right)\mathbf{S}\right)\right],
\end{equation}
\begin{equation}
\label{eq:derivation of q_2}  
\frac{\partial\mathbf{q}
}{\partial \boldsymbol{\varrho}^T}=\text{vec}\left(\mathbf{D}\mathbf{G}\left(\mathbf{x}_{\text{p}}\right)\mathbf{S}\right)\left[1,j\right],
\end{equation}
where $\dot{\mathbf{D}}_{u}=\frac{\partial\mathbf{D}
}{\partial u}=\frac{\partial\mathbf{a}\left(\theta_{\text{b}}\right)
}{\partial u}\mathbf{h}_{\text{s}}^H\left(\mathbf{x}_{\text{p}}\right)+\mathbf{a}\left(\theta_{\text{b}}\right)\frac{\partial\mathbf{h}_{\text{s}}^H(\mathbf{x}_{\text{p}})
}{\partial u},\forall u\in\{x_{\text{s}},y
_{\text{s}}\}$. 
%Specifically, the partial derivatives of $\mathbf{a}\left(\theta_{\text{b}}\right)$ with respect to $x_{\text{s}}$ and $y_{\text{s}}$ are given by
%\begin{equation}
%\begin{aligned}
%\label{eq:derivation of a_x}  
%&\frac{\partial\mathbf{a}\left(\theta_{\text{b}}\right)
%}{\partial x_{\text{s}}}=j\beta_0\frac{\lambda x_{\text{s}}y_{\text{s}}}{2(x_{\text{s}}^2+y_{\text{s}}^2)^{\frac{3}{2}}}\\&\times\left[0,e^{-j\beta_0\frac{\lambda\sin\theta_{\text{b}}}{2}},\cdots,\left(K-1\right)e^{-j\beta_0\frac{\left(K-1\right)\lambda\sin\theta_{\text{b}}}{2}}\right]^T, 
%\end{aligned}
%\end{equation}
%\begin{equation}
%\begin{aligned}
%\label{eq:derivation of a_y}  
%&\frac{\partial\mathbf{a}\left(\theta_{\text{b}}\right)
%}{\partial y_{\text{s}}}=-j\beta_0\frac{\lambda x_{\text{s}}^2}{2(x_{\text{s}}^2+y_{\text{s}}^2)^{\frac{3}{2}}}\\&\times\left[0,e^{-j\beta_0\frac{\lambda\sin\theta_{\text{b}}}{2}},\cdots,(K-1)e^{-j\beta_0\frac{(K-1)\lambda\sin\theta_{\text{b}}}{2}}\right]^T.
%\end{aligned}
%\end{equation}

With sufficiently large sensing frame length $T$, the covariance matrix $\mathbf{R}_{\text{s}}$ can be approximated as $\mathbf{R}_{\text{s}}\approx\frac{1}{T}\mathbf{S}\mathbf{S}^H$. For any $u,v\in\{x_{\text{s}},y_{\text{s}}\}$, the entry of the matrix $\mathbf{J}_{11}$ is given by 
\begin{equation}
\label{eq:J11}
J_{uv} = \frac{2|\varrho|^2T}{\sigma_{\text{s}}^2} \text{Re}\left\{
\mathrm{tr} \left(\dot{\mathbf{D}}_{v}\mathbf{G}\left(\mathbf{x}_{\text{p}}\right)\mathbf{R}_{\text{s}}\mathbf{G}^H\left(\mathbf{x}_{\text{p}}\right)\dot{\mathbf{D}}_{u}^H\right)\right\}.
\end{equation}
 The matrices 
 $\mathbf{J}_{12}$ and $\mathbf{J}_{22}$ are given by
\begin{equation}
\begin{aligned}
\label{eq:J12}
&\mathbf{J}_{12} = \\&\frac{2T}{\sigma_{\text{s}}^2}\text{Re}
\left\{\begin{bmatrix}\varrho^*\mathrm{tr}\left(\mathbf{D}\mathbf{G}\left(\mathbf{x}_{\text{p}}\right)\mathbf{R}_{\text{s}}\mathbf{G}^H\left(\mathbf{x}_{\text{p}}\right)\dot{\mathbf{D}}_{x_{\text{s}}}^H\right)\\\varrho^*\mathrm{tr}\left(\mathbf{D}\mathbf{G}\left(\mathbf{x}_{\text{p}}\right)\mathbf{R}_{\text{s}}\mathbf{G}^H\left(\mathbf{x}_{\text{p}}\right)\dot{\mathbf{D}}_{y_{\text{s}}}^H\right)
\end{bmatrix}
\begin{bmatrix}
1, j
\end{bmatrix}
\right\},
\end{aligned}
\end{equation}
\begin{equation}
\label{eq:J22}
\mathbf{J}_{22} = \frac{2T}{\sigma_{\text{s}}^2} \mathbf{I}_2\text{Re}\left\{
\mathrm{tr} \left(\mathbf{D}\mathbf{G}\left(\mathbf{x}_{\text{p}}\right)\mathbf{R}_{\text{s}}\mathbf{G}^H\left(\mathbf{x}_{\text{p}}\right)\mathbf{D}^H\right)\right\}.
\end{equation}

\end{appendices}

\bibliographystyle{IEEEtran}
\bibliography{sample}

@STRING{IEEE_J_COML       = "{IEEE} Commun. Lett."}

@STRING{IEEE_J_WCOML       = "{IEEE} Wireless Commun. Lett."}

@STRING{IEEE_J_JSAC       = "{IEEE} J. Sel. Areas Commun."}

@STRING{IEEE_J_COM        = "{IEEE} Trans. Commun."}

@STRING{IEEE_J_SP         = "{IEEE} Trans. Signal Process."}

@STRING{IEEE_J_WCOM       = "{IEEE} Trans. Wireless Commun."}

@STRING{IEEE_J_VT         = "{IEEE} Trans. Veh. Technol."}

@STRING{IEEE_WM_COM        = "{IEEE} Wireless Commun."}

@STRING{IEEE_J_STSP  = "{IEEE} J. Sel. Topics Signal Process."}

@article{pinching2,
  author={Yang, Zheng and Wang, Ning and Sun, Yanshi and Ding, Zhiguo and Schober, Robert and Karagiannidis, George K. and Wong, Vincent W.S. and Dobre, Octavia A.},
  journal=IEEE_WM_COM, 
  title={Pinching Antennas: {Principles}, Applications and Challenges}, 
  year={2026},
  volume={33},
  number={2},
  pages={175-184},
  month={Apr.}}

@ARTICLE{pinching-arraygain,
  author={Ouyang, Chongjun and Wang, Zhaolin and Liu, Yuanwei and Ding, Zhiguo},
  journal=IEEE_J_COML, 
  title={Array Gain for Pinching-Antenna Systems {(PASS)}}, 
  year={2025},
  volume={29},
  number={6},
  pages={1471-1475},
  month={Jun.}}

@article{pinching-uplink,
  author={Tegos, Sotiris A. and Diamantoulakis, Panagiotis D. and Ding, Zhiguo and Karagiannidis, George K.},
  journal=IEEE_J_WCOML, 
  title={Minimum Data Rate Maximization for Uplink Pinching-Antenna Systems}, 
  year={2025},
  volume={14},
  number={5},
  pages={1516-1520},
  month={May}}

@ARTICLE{pinching-noma,
  author={Xu, Yanqing and Ding, Zhiguo and Cai, Donghong and Wong, Vincent W.S.},
  journal=IEEE_J_COM, 
  title={{QoS}-Aware {NOMA} Design for Downlink Pinching-Antenna Systems}, 
  year={2025},
  volume={73},
  number={12},
  pages={13611-13625},
  month={Dec.}}

@ARTICLE{pinching-PLS,
  author={Wang, Kaidi and Ding, Zhiguo and Al-Dhahir, Naofal},
  journal=IEEE_J_WCOML,
  title={Pinching-Antenna Systems for Physical Layer Security}, 
  year={2026},
  volume={15},
  number={},
  pages={260-264}}

@ARTICLE{CRB,
  author={Liu, Fan and Liu, Ya-Feng and Li, Ang and Masouros, Christos and Eldar, Yonina C.},
  journal=IEEE_J_SP, 
  title={Cramér-Rao Bound Optimization for Joint Radar-Communication Beamforming}, 
  year={2022},
  volume={70},
  number={},
  pages={240-253},
  }

@ARTICLE{Massive-mimo,
  author={Yu, Xianghao and Shen, Juei-Chin and Zhang, Jun and Letaief, Khaled B.},
  journal=IEEE_J_STSP, 
  title={Alternating Minimization Algorithms for Hybrid Precoding in Millimeter Wave {MIMO} Systems}, 
  year={2016},
  volume={10},
  number={3},
  pages={485-500},
  month={Apr.}}

@misc{cvx,
  author       = {Michael Grant and Stephen Boyd},
  title        = {{CVX}: Matlab Software for Disciplined Convex Programming, version 2.1},
  howpublished = {\url{https://cvxr.com/cvx}},
  month        = mar,
  year         = 2014
}

@ARTICLE{tutorial-lyw,
  author={Liu, Yuanwei and Jiang, Hao and Xu, Xiaoxia and Wang, Zhaolin and Guo, Jia and Ouyang, Chongjun and Mu, Xidong and Ding, Zhiguo and Nallanathan, Arumugam and Karagiannidis, George K. and Schober, Robert},
  journal=IEEE_J_COM, 
  title={Pinching-Antenna Systems {(PASS)}: {A} Tutorial}, 
  year={2026},
  volume={74},
  number={},
  pages={4881-4918}}

@ARTICLE{modeling-wzl,
  author={Wang, Zhaolin and Ouyang, Chongjun and Mu, Xidong and Liu, Yuanwei and Ding, Zhiguo},
  journal=IEEE_J_COM, 
  title={Modeling and Beamforming Optimization for Pinching-Antenna Systems}, 
  year={2025},
  volume={73},
  number={12},
  pages={13904-13919},
  month={Dec.}}

@ARTICLE{Related-ISAC-LHC,
  author={Li, Haochen and Zhong, Ruikang and Pan, Zhiwen and Dong, Chao and Lei, Jiayi and Liu, Yuanwei},
  journal=IEEE_J_WCOM, 
  title={Pinching Antenna Systems for Integrated Sensing and Communications}, 
  year={2026},
  volume={25},
  number={},
  pages={13416-13429},
  }

@article{Mohanty2017,
  author  = {Aseema Mohanty and Mian Zhang and Avik Dutt
             and Sven Ramelow and Paulo Nussenzveig and Michal Lipson},
  title   = {Quantum Interference Between Transverse Spatial Waveguide Modes},
  journal = {Nat. Commun.},
  volume  = {8},
  pages   = {14010},
  month   = {Jan.},
  year    = {2017}
}

@book{Huang1984CMT,
  author    = {H. Huang},
  title     = {Coupled Mode Theory: As Applied to Microwave 
               and Optical Transmission},
  address   = {Utrecht, The Netherlands},
  publisher = {VNU Science Press},
  year      = {1984},
  isbn      = {978-90-6764-033-6}
}

@article{Huang2009CMT,
  author  = {W.-P. Huang and J. Mu},
  title   = {Complex Coupled-Mode Theory for Optical Waveguides},
  journal = {Opt. Express},
  volume  = {17},
  number  = {21},
  pages   = {19134--19152},
  year    = {2009},
  doi     = {10.1364/OE.17.019134}
}

@ARTICLE{Clercparticle2002,
  author={Clerc, M. and Kennedy, J.},
  journal={IEEE Trans. Evol. Comput.}, 
  title={The particle swarm-explosion, stability, and convergence in a multidimensional complex space}, 
  year={2002},
  month={Feb.},
  volume={6},
  number={1},
  pages={58-73},
  doi={10.1109/4235.985692}}

@book{matching1,
    author = {Roth, A. E and Sotomayor, M. A. O},
    title = {Two-Sided Matching: a study in game-theoretic modeling and analysis},
    publisher = {Cambridge, U.K.: Cambridge Univ. Press},
    year = {1992}
}

@article{Yariv1973CMT,
  author  = {A. Yariv},
  title   = {Coupled-Mode Theory for Guided-Wave Optics},
  journal = {{IEEE} J. Quantum Electron.},
  volume  = {9},
  number  = {9},
  pages   = {919--933},
  month   = sep,
  year    = {1973},
  doi     = {10.1109/JQE.1973.1077767}
}

@book{Skolnik2001Radar,
  author    = {Merrill I. Skolnik},
  title     = {Introduction to Radar Systems},
  edition   = {3rd},
  publisher = {McGraw-Hill},
  address   = {New York, NY, USA},
  year      = {2001}
}

@article{xiaoxia-multimode,
  title={Multi-Mode Pinching Antenna Systems Enabled Multi-User Communications},
  author={Xu, Xiaoxia and Mu, Xidong and Liu, Yuanwei and Nallanathan, Arumugam},
  journal={arXiv preprint  arXiv:2601.20780},
  year={2026}
}

@article{shaodan-multimode,
  title={Multi-Mode Pinching-Antenna Systems:
Polarization-Aware Full-Wave Modeling and
Optimizatio},
  author={D. Wei and R. Zhang and Y. Shao and F. Hou and S. Ma},
  journal={arXiv preprint  arXiv:2604.01778},
  year={2026}
}

@ARTICLE{10945421,
  author={Ding, Zhiguo and Schober, Robert and Vincent Poor, H.},
  journal={{IEEE Trans. Commun.}}, 
  title={Flexible-Antenna Systems: A Pinching-Antenna Perspective}, 
  year={2025},
  volume={73},
  number={10},
  pages={9236-9253},
  month={Oct.},
  doi={10.1109/TCOMM.2025.3555866}}

@ARTICLE{9264694,
  author={Wong, Kai-Kit and Shojaeifard, Arman and Tong, Kin-Fai and Zhang, Yangyang},
  journal=IEEE_J_WCOM, 
  title={Fluid Antenna Systems}, 
  year={2021},
  volume={20},
  number={3},
  pages={1950-1962},
  month={Mar.}}

@ARTICLE{10318061,
  author={Zhu, Lipeng and Ma, Wenyan and Zhang, Rui},
  journal=IEEE_J_WCOM, 
  title={Modeling and Performance Analysis for Movable Antenna Enabled Wireless Communications}, 
  year={2024},
  volume={23},
  number={6},
  pages={6234-6250},
  month={June}}

@ARTICLE{marco,
  author={Di Renzo, Marco and Zappone, Alessio and Debbah, Merouane and Alouini, Mohamed-Slim and Yuen, Chau and de Rosny, Julien and Tretyakov, Sergei},
  journal=IEEE_J_JSAC, 
  title={Smart Radio Environments Empowered by Reconfigurable Intelligent Surfaces: How It Works, State of Research, and The Road Ahead}, 
  year={2020},
  volume={38},
  number={11},
  pages={2450-2525},
  month={Nov.}}

@ARTICLE{11202577,
  author={Wang, Zhaolin and Ouyang, Chongjun and Mu, Xidong and Liu, Yuanwei and Ding, Zhiguo},
  journal={{IEEE Trans. Commun.}}, 
  title={Modeling and Beamforming Optimization for Pinching-Antenna Systems}, 
  year={2025},
  volume={73},
  number={12},
  pages={13904-13919},
  month={Dec.},
  doi={10.1109/TCOMM.2025.3621049}}

@ARTICLE{11263923,
  author={Xu, Xiaoxia and Mu, Xidong and Wang, Zhaolin and Liu, Yuanwei and Nallanathan, Arumugam},
  journal={{IEEE Trans. Commun.}}, 
  title={Pinching-Antenna Systems ({PASS}): Power Radiation Model and Optimal Beamforming Design}, 
  year={2026},
  volume={74},
  number={},
  pages={2160-2175},
  month={Nov,},
  doi={10.1109/TCOMM.2025.3636083}}

@ARTICLE{11488474,
  author={Ding, Zhiguo and Schober, Robert and Poor, H. Vincent},
  journal=IEEE_J_WCOM, 
  title={Environment Division Multiple Access ({EDMA}): A Feasibility Study via Pinching Antennas}, 
  year={2026},
  volume={25},
  number={},
  pages={15675-15691},
  month={Apr.},
  doi={10.1109/TWC.2026.3683028}}

@ARTICLE{11212813,
  author={Mao, Weihao and Lu, Yang and Xu, Yanqing and Ai, Bo and Dobre, Octavia A. and Niyato, Dusit},
  journal=IEEE_J_WCOM, 
  title={Multi-Waveguide Pinching Antennas for {ISAC}}, 
  year={2025},
  volume={25},
  number={},
  pages={5846-5858},
  month={Oct.},
  doi={10.1109/TWC.2025.3621316}}

@ARTICLE{11599010,
  author={Yu, Cong and Ai, Bo and Gan, Xu and Liu, Yuanwei and Shi, Guowei and Chen, Wei},
  journal=IEEE_J_WCOM, 
  title={Equivalent Radiation Control for {ISAC} in Pinching Antenna Systems: A Discrete Activation Framework}, 
  year={2026},
  volume={25},
  number={},
  pages={19825-19841},
  month={Jul.},
  doi={10.1109/TWC.2026.3707879}}

@ARTICLE{11288076,
  author={Hu, Yanglin and Zhang, Tiankui and Xu, Xiaoxia and Liu, Yuanwei},
  journal=IEEE_J_VT, 
  title={Pinching Antenna-Enabled Integrated Sensing and Communication for Low-Altitude {UAV}}, 
  year={2026},
  volume={75},
  number={6},
  pages={11728-11733},
  month={June},
  doi={10.1109/TVT.2025.3641978}}

@article{Marcatili1969,
  author  = {E. A. J. Marcatili},
  title   = {Dielectric Rectangular Waveguide and Directional Coupler for Integrated Optics},
  journal = {Bell Syst. Tech. J.},
  volume  = {48},
  number  = {7},
  pages   = {2071--2102},
  month   = {Sep.},
  year    = {1969}
}

\end{document}